%% file: arXiv.tex
\documentclass[11pt,letterpaper]{article}
\usepackage[utf8]{inputenc}
\usepackage{graphicx}
\usepackage[letterpaper,margin=1in]{geometry}
\usepackage{amsmath,amssymb,amsthm,verbatim}
\usepackage[ruled,vlined]{algorithm2e}
\usepackage{algorithmic}
\usepackage[colorlinks=true, linkcolor=black, citecolor=black]{hyperref}
\usepackage{cleveref} %enables us to use cross-referencing
\usepackage{thmtools, thm-restate}

\usepackage[dvipsnames]{xcolor}
\usepackage{tikz,pgfplots}
\pgfplotsset{compat=1.18} 
\usetikzlibrary{patterns}
\usetikzlibrary{arrows.meta}
\newtheorem{theorem}{Theorem}
\newtheorem{lemma}[theorem]{Lemma}
\newtheorem{claim}[theorem]{Claim}
\newtheorem{corollary}[theorem]{Corollary}

\newtheorem{observation}[theorem]{Observation}
\theoremstyle{definition}
\newtheorem{definition}[theorem]{Definition}

\def\DEBUG{true} % set this to be false to remove comments
\ifdefined\DEBUG

  \def\rem#1{{\marginpar{\raggedright\scriptsize #1}}}
  \newcommand{\adir}[1]{\rem{\textcolor{Red}{$\bullet$ Adithya: #1}}}
  \newcommand{\ankr}[1]{\rem{\textcolor{Green}{$\bullet$ Ankit: #1}}}
  \newcommand{\ashr}[1]{\rem{\textcolor{Blue}{$\bullet$ Ashish: #1}}}
\else

  \newcommand{\adir}[1]{}
  \newcommand{\ankr}[1]{}
  \newcommand{\ashr}[1]{}
\fi

\title{Weighted $k$-Server Admits an Exponentially Competitive Algorithm}
\author{
  Adithya Bijoy\thanks{National University of Singapore. This work done while the author was affiliated to Indian Institute of Technology, Delhi, India, and was partially supported by the CSE Research Acceleration Fund of IIT Delhi.} \and
  Ankit Mondal\thanks{Indian Institute of Technology, Delhi, India}  \and
  Ashish Chiplunkar\thanks{Indian Institute of Technology, Delhi, India, \url{https://www.cse.iitd.ac.in/\~ashishc/}}
}
\date{}

\begin{document}
\maketitle

\begin{abstract}
The weighted $k$-server is a variant of the $k$-server problem, where the cost of moving a server is the server's weight times the distance through which it moves. The problem is famous for its intriguing properties and for evading standard techniques for designing and analyzing online algorithms. Even on uniform metric spaces with sufficiently many points, the deterministic competitive ratio of weighted $k$-server is known to increase doubly exponentially with respect to $k$, while the behavior of its randomized competitive ratio is not fully understood. Specifically, no upper bound better than doubly exponential is known, while the best known lower bound is singly exponential in $k$. In this paper, we close the exponential gap between these bounds by giving an $\exp(\mathcal{O}(k^2))$-competitive randomized online algorithm for the weighted $k$-server problem on uniform metrics, thus breaking the doubly exponential barrier for deterministic algorithms for the first time. This is achieved by a recursively defined notion of a \textit{phase} which, on the one hand, forces a lower bound on the cost of any offline solution, while, on the other hand, also admits a randomized online algorithm with bounded expected cost. The algorithm is also recursive; it involves running several algorithms virtually and in parallel and following the decisions of one of them in a random order. We also show that our techniques can be lifted to construct an $\exp(\mathcal{O}(k^2))$-competitive randomized online algorithm for the generalized $k$-server problem on weighted uniform metrics.
\end{abstract}

\section{Introduction}\label{sec_intro}

The $k$-server problem of Manasse, McGeoch, and Sleator~\cite{ManasseMS_STOC88} is considered to be one of the cornerstone problems in the domain of online computation. The problem concerns serving requests on an underlying metric space using $k$ identical mobile servers. Specifically, an adversary presents a sequence of requests, one request at a time, where each request is a point in the metric space. In response, an online algorithm must move one of its $k$ servers to the requested point (unless a server is already located there). The cost of moving a server is the distance through which the server moves, and the objective is to minimize the total cost. 
Since the algorithm must serve a request before it gets to see the next one, it is unreasonable to hope that the algorithm outputs the cheapest solution to the instance. The performance of an online algorithm is measured using the framework of \textit{competitive analysis} introduced by Sleator and Tarjan~\cite{SleatorT85}. Informally, the \textit{competitive ratio} of an online algorithm is the worst case ratio of the (expected) cost of the algorithm's solution to the cost of an optimal offline solution. While analyzing the competitive ratio of randomized algorithms it is customary to assume that the algorithm competes against an \textit{oblivious adversary} who does know the algorithm but can't access the outcomes of the randomness used by the algorithm, and thus, can't necessarily infer the positions of the algorithm's servers. The competitive ratio of an online problem is the best competitive ratio that can be attained by any online algorithm for the problem.

The competitive ratio of deterministic algorithms for the $k$-server problem is reasonably well understood. Already in their seminal paper, Manasse, McGeoch, and Sleator~\cite{ManasseMS_STOC88} proved that no deterministic algorithm for the $k$-server problem can have competitive ratio less than $k$ on any metric space having more than $k$ points. They also posed the $k$-server conjecture, which guarantees the existence of a deterministic $k$-competitive algorithm for the $k$-server problem on arbitrary metric spaces. Fiat, Rabani, and Ravid~\cite{FiatRR_FOCS90} were the first to design an algorithm for the $k$-server problem with competitive ratio bounded by a function of $k$ only. The currently best known algorithm for $k$-server is the so-called \textit{work function algorithm} by Koutsoupias and Papadimitriou~\cite{KoutsoupiasP_JACM95}, which is known to have competitive ratio of at most $2k-1$ on every metric space. In the case of randomized algorithms, the recent breakthrough by Bubeck, Coester, and Rabani~\cite{BubeckCR_STOC23} shows a competitive ratio lower bound of $\Omega(\log^2k)$ on an appropriately designed metric space, and $\Omega(\log k)$ on every metric space with more than $k$ points. It is noteworthy that no upper bound on the randomized competitive ratio better than $2k-1$ is known. The $k$-server problem on \textit{uniform metrics} (where any two distinct points are a unit distance apart) is particularly interesting because it is identical to the paging problem. In this case, the deterministic competitive ratio is exactly $k$ \cite{SleatorT85}, while the randomized competitive ratio is exactly $h(k)=1+1/2+\cdots+1/k$ \cite{FiatKLMSY_JAlg91,AchlioptasCN_TCS00}.

\subsection*{Weighted $k$-server}

A natural generalization of the $k$-server defined by Newberg~\cite{Newberg91} is where servers are non-identical -- they have \textit{weights} $w_1,\ldots,w_k$ and the cost of moving a server of weight $w$ through a distance $d$ is $w\cdot d$. In this case, it is reasonably easy to see that a $c$-competitive $k$-server algorithm is also a $c\cdot(w_{\max}/w_{\min})$-competitive weighted $k$-server algorithm, where $w_{\max}$ and $w_{\min}$ are the maximum and minimum weights respectively. However, the challenge here is to design algorithms having competitive ratios that depend only on $k$ and not on the weights.

The weighted $k$-server problem turns out to be notoriously hard on arbitrary metric spaces in two ways. First, the problem seems to evade standard techniques to such an extent that no competitive algorithm is known for arbitrary metrics for any $k>2$. Second, even when competitive ratio bounds are known, they are much worse than the bounds in the unweighted case for the same $k$ and the same metric space. For example, the best-known upper bound on the competitive ratio for $k=2$ on arbitrary metrics is 879~\cite{SittersS_JACM06}.

Due to its difficulty, it becomes prudent to understand the competitive ratio of weighted $k$-server on \textit{simple} metric spaces, and in particular, the uniform metric space. On the uniform metric space, the weighted $k$-server problem is identical to a variant of the paging problem where the cost of page replacement is determined by the memory location where the replacement takes place. Fiat and Ricklin~\cite{FiatR_TCS94} designed a deterministic algorithm for weighted $k$-server on uniform metrics whose competitive ratio increases doubly exponentially with respect to $k$. More than two decades later, Bansal, Eli\'{a}s, and Koumoutsos~\cite{BansalEK_FOCS17} proved a doubly exponential lower bound on the competitive ratio of any deterministic algorithm. Contrast this with the unweighted case where the deterministic competitive ratio is only $k$.

Until recently, the understanding of the randomized competitive ratio of weighted $k$-server on uniform metrics was embarrassing. The best upper bound known so far is doubly exponential in $k$, resulting from a deterministic algorithm~\cite{FiatR_TCS94}, and from a randomized memoryless algorithm~\cite{ChiplunkarV_TAlg20} against a stronger form of adversary called the adaptive online adversary (refer to~\cite{Ben-DavidBKTW94} for the definition). In fact, the doubly exponential dependence is unavoidable under each of those restrictions. On the other hand, no better lower bound than the $\Omega(\log k)$ bound in the unweighted case was known, thus, leaving a triply exponential gap between the upper and lower bounds. Recently, Ayyadevara and Chiplunkar~\cite{AyyadevaraC_ESA21} achieved a breakthrough by improving the lower bound to $\Omega(2^{k})$, reducing the gap to singly exponential. Very recently, Ayyadevara, Chiplunkar, and Sharma~\cite{AyyadevaraCS_FSTTCS24} took a decomposition approach to the problem, where they defined two relaxations of weighed $k$-server with the property that the product of their competitive ratios is an upper bound on the product of weighted $k$-server on uniform metrics. They proved that the exponential lower bound of weighted $k$-server applies to one of the relaxations and gave an $\exp(\mathcal{O}(k^2))$ algorithm for the same. However, because the other relaxation was left unsolved, the exponential gap remains.

In this paper, we close the exponential gap and prove that the randomized competitive ratio of weighted $k$-server on uniform metrics indeed behaves singly exponentially with respect to $k$. Specifically, we prove,

\begin{restatable}{theorem}{main}\label{thm_main}
There exists an $\exp(\mathcal{O}(k^2))$-competitive randomized algorithm for the weighted $k$-server problem on uniform metrics.
\end{restatable}

\subsection*{Generalized $k$-server}

The generalized $k$-server problem on weighted uniform metrics is a generalization of the weighted $k$-server problem on uniform metrics. Here, we have $k$ pairwise disjoint uniform metric spaces $M_1,M_2,\ldots,M_k$, and a server $s_{\ell}$ that moves in each $M_{\ell}$. Distinct points in each metric space $M_{\ell}$ are separated by a distance $w_{\ell}$. Each request is a tuple $r \in M_1\times M_2 \times \cdots \times M_k$ of $k$ points, where we denote the $\ell$'th component of the tuple by $r[\ell]$. The algorithm is required to satisfy the request by ensuring that at least one server $s_{\ell}$ is located at the corresponding point $r[\ell]$. A movement of $s_{\ell}$ results in a cost of $w_{\ell}$. 

It is easy to see that the weighted $k$-server problem can be reduced to the generalized $k$-server problem. Let $M'$ be the underlying uniform metric space for the weighted $k$-server problem, and let $w_1,\ldots,w_k$ be the weights of servers. Make $k$ copies $M_1,\ldots,M_k$ of $M'$ and let $w_{\ell}$ be the pairwise distance between points in each $M_{\ell}$. Simulate a weighted $k$-server request to a point $p\in M$ by a generalized $k$-server request consisting of the copies of $p$ in all $M_{\ell}$'s.

Although generalized $k$-server might seem to be strictly harder than weighted $k$-server and no reduction from the former to the latter is known, algorithms originally designed for weighted $k$-server have been lifted to get algorithms for generalized $k$-server with the same competitive ratio and minor changes to proofs~\cite{BansalEKN_TALG23,ChiplunkarV_TAlg20}. In the same vein, we prove that our algorithm and analysis can be lifted to establish the following result, whose proof sketch is presented in Appendix~\ref{app_gks}.

\begin{restatable}{theorem}{thmgks}\label{thm_gks}
There exists an $\exp(\mathcal{O}(k^2))$-competitive randomized algorithm for the generalized $k$-server problem on weighted uniform metrics.
\end{restatable}

\subsection*{Our techniques} 

On a very high level, the idea behind the analysis of our algorithm is rather straightforward: we show that any request sequence can be split into \textit{phases}, where each phase forces a lower bound on the optimum cost while also allowing an online algorithm to serve its requests with bounded expected cost. However, the definition of a phase is tricky and involves induction. Assume that we number servers in non-decreasing order of weight so that $w_1\leq\cdots\leq w_k$. Intuitively, for any $\ell\in\{1,\ldots,k\}$, and a set $H$ of at most $k-\ell$ points in the metric space, an $(\ell,H)$-phase is a request sequence $\rho$ that simultaneously satisfies the following two properties.
\begin{itemize}
\item If we start with the heaviest $k-\ell$ servers covering points in $H$, then it is possible to serve all requests in $\rho$ in an online manner with expected cost at most $c_{\ell}\cdot w_{\ell}$, using only the lightest $\ell$ servers.
\item If we start with the heaviest $k-\ell$ servers not covering any point from a set of important points associated with $\rho$, then any solution, not necessarily online, that serves $\rho$ must pay a cost of at least $w_{\ell}/c'_{\ell}$.
\end{itemize}
Here, $c_{\ell}$ and $c'_{\ell}$ are constants that depend on $\ell$ only. We show how we can build up an $(\ell,H)$-phase from several $(\ell-1,H')$-phases and unambiguously parse a request sequence as a concatenation of $(k,\emptyset)$-phases (and an incomplete $(k,\emptyset)$-phase). As a result, we establish a competitive ratio of $c_k\cdot c'_k$, which turns out to be $\exp(\mathcal{O}(k^2))$.

Our online algorithm to serve an $(\ell,H)$-phase uses the following idea. Suppose it starts out from a configuration that covers all points in $H$ with the heaviest $k-\ell$ servers. The algorithm never moves those servers, so it can ignore requests to points in $H$. The first half of the phase consists of several $(\ell-1,H)$-phases. The algorithm serves these phases recursively using only the lightest $\ell-1$ servers, and while doing so, identifies a \textit{critical set} $S$ of points. Informally, these points are the best candidates for the adversary to keep its $\ell$'th lightest server during the phase. The behavior of our algorithm in the second half of a phase is inspired by the so-called \textit{randomized-min operator} resulting from Theorem 11 of~\cite{FiatFKRRV_SIAMJC98}. For the second half of the phase, our algorithm, in its imagination and in parallel for each $p\in S$, serves requests by keeping the $\ell$'th lightest server at $p$ and recursively running algorithms to serve a fixed number of $(\ell-1,H\cup\{p\})$-phases. In reality, the algorithm chooses one of those runs uniformly at random and follows it until it terminates. Upon termination of that imaginary run, the algorithm switches to another uniformly random imaginary run that has not terminated. This continues until all the imaginary runs terminate, and by this time, our algorithm is done serving the phase. This process serves the purpose of weeding out bad choices for the location of the $\ell$'th lightest server -- indeed, if the sequence of requests of the fixed number of $(\ell-1,H\cup\{p\})$-phases is too short, it is very unlikely that the corresponding run is ever chosen by the algorithm to follow.

Our idea might appear to be similar to that of Fiat and Ricklin~\cite{FiatR_TCS94} behind their deterministic algorithm, but the similarity is only superficial. Since Fiat and Ricklin analyze a deterministic algorithm, they can assume without loss of generality that the adversary never requests a point which is already occupied by any of the algorithm's servers. This enables them to define phases based on the number of requests only, and claim an exact value of the cost paid by their algorithm in every phase. None of this works as soon as the algorithm is randomized and the adversary is oblivious. Therefore, we have to meticulously define phases and their critical sets, making the definitions more involved than those of Fiat and Ricklin. The other difference is in the second part of a phase. Here, Fiat and Ricklin's algorithm visits all points in the critical set, one at a time, in arbitrary order. In contrast, we simulate all these choices in parallel to weed out the worse ones and try out only those choices that survive for a sufficiently long duration.

\subsection*{Organization of the paper}

In Section~\ref{sec_prelims}, we formally define the weighted $k$-server problem and state assumptions that we make without loss of generality to simplify the presentation of our algorithm and its analysis. Section~\ref{sec_phase} focuses on defining phases and several attributes of phases, including the request sequence associated with a phase and its critical set. In Section~\ref{sec_opt}, we show lower bounds on the optimal cost of serving request sequences associated with phases. We present the recursive construction of an online algorithm to serve requests that constitute a phase in Section~\ref{sec_strategy}. Finally, we put everything together and prove a competitive guarantee in Section~\ref{sec_alg}. We present proofs skipped from these sections in Appendix~\ref{app_missingproofs}, while Appendix~\ref{app_gks} is devoted to listing the changes that need to be made to lift the algorithm and its analysis to generalized $k$-server on weighted uniform metrics.

\section{Preliminaries}\label{sec_prelims}

Consider a set $M$, referred to as the \textit{set of points} henceforth. The uniform metric on $M$ is the function that assigns a unit \textit{distance} between any pair of distinct points in $M$, and a zero distance between a point and itself. The weighted $k$-server problem on the uniform metric on $M$ is specified by $k$ non-negative real numbers $w_1\leq\cdots\leq w_k$, which are the weights of $k$ servers occupying points in $M$, and the initial positions of all servers. In each round, an adversary specifies a request on a point in $M$, and it must be served by moving one of the servers to the requested point, unless a server already occupies that point. The cost of a movement (if any) is the weight of the server moved. The objective is to minimize the total cost of serving a finite sequence $\rho$ of points in $M$. We say that a randomized online algorithm $\mathcal{A}$ is $c$-competitive for a function $c$ of $k$ (alone) if there exists $c'$ (possibly depending on weights) such that for every sequence $\rho$ of points in $M$, we have
\[\text{ALG}(\rho)\leq c\cdot\text{OPT}(\rho)+c'\text{,}\]
where $\text{ALG}(\rho)$ denotes the expected cost of $\mathcal{A}$ on $\rho$, and $\text{OPT}(\rho)$ denote the minimum cost of serving requests $\rho$ (in hindsight). Note that neither $c$ nor $c'$ is allowed to depend on $\rho$. The goal of this paper is to give an $\exp(\mathcal{O}(k^2))$-competitive randomized online algorithm for weighted $k$-server on uniform metrics.

We assume that for each $i\in\{2,\ldots,k\}$, $w_i$ is a multiple of $w_{i-1}$ greater than $w_{i-1}$. If this condition does not hold, we can successively round up each $w_i$ to a multiple of $w_{i-1}$. This distorts each $w_i$ by at most a factor $2^{i-1}$, resulting in a loss of at most $2^{k-1}$ in the competitive ratio (see Lemma~\ref{lem_distort} in Appendix~\ref{appendix_prelims}). Since we are targeting an $\exp(\mathcal{O}(k^2))$ bound on the competitive ratio, the $2^{k-1}$ factor loss can be ignored, and hence, there is no loss of generality in making the assumption.

We treat a request sequence as a string over the alphabet of points in the metric space. Here are some related definitions and notation.
\begin{definition}\label{def_concat}
The concatenation of strings $\rho_1$ and $\rho_2$ is written as $\rho_1\rho_2$. If the string $\rho_1$ is a prefix of the string $\rho_2$, then we use $\rho_1\backslash\rho_2$ to denote the unique string $\rho_3$ such that $\rho_2=\rho_1\rho_3$, that is, the string that remains on removing the prefix $\rho_1$ from $\rho_2$. We call a finite set $R$ of strings a \textit{prefix chain} if for every $\rho_1,\rho_2\in R$ one of $\rho_1$, $\rho_2$ is a prefix of the other, or equivalently, $R$ contains a unique maximum length string and every string in $R$ is a prefix of that maximum length string.
\end{definition}

We will be using non-negative vectors in the real vector space with dimensions being the points in the metric space. Here are some related definitions and notation.
\begin{definition}\label{def_one_hot}
For a point $p\in M$, $u_p$ denotes the vector that has coordinate value $1$ in the direction $p$ and $0$ in all other directions. If $v$ is a non-negative vector in this vector space and $d\in\mathbb{N}$, then $\text{top}_d(v)$ denotes the set of $d$ points that have the $d$ maximum coordinates in $v$, breaking ties using some fixed total order on the set of points $M$. By $|v|$, we denote the $1$-norm of $v$, which equals the sum of entries of $v$.
\end{definition}

Throughout this paper, we use constants $d_1,\ldots,d_k$ defined as follows.
\begin{definition}\label{def_dl}
$d_{\ell}=2^{5^{\ell-1}-1}$.
\end{definition}

\section{Phases and Multiphases}\label{sec_phase}

The notion of a \textit{phase} and a \textit{multiphase} is central to the design and analysis of our algorithm for weighted $k$-server. Informally, a phase represents a request sequence which is hard enough for an offline solution, while being easy enough for an online algorithm. More specifically, for $\ell\in\{1,\ldots,k\}$ and a subset of $H$ at most $k-\ell$ points, an $(\ell,H)$-phase is a recursively defined structure that exists for a request sequence only if the sequence can be served in an online manner with a small enough cost if we start out with the heaviest $k-\ell$ servers covering all points in $H$, while any offline solution necessarily incurs sufficiently large cost, unless some of the heavier servers cover some crucial points initially. These properties are formalized and rigorously proven in the next two sections. An $(\ell,H)$-multiphase can be simply thought of as a sequence of $w_{\ell+1}/w_{\ell}$ $(\ell,H)$-phases. The terms $(\ell,H)$-phase and $(\ell,H)$-multiphase, their request sequence, their demand vector, and their critical set are defined inductively as follows. Appendix~\ref{appendix_example} illustrates these definitions with an example.

\begin{definition}\label{def_phase}
Let $\ell\in\{1,\ldots,k\}$ and let $H$ be a set of at most $k-\ell$ points from the metric space. 

A \textit{$(1,H)$-phase} is any string $\rho$ over the set $H\cup\{p\}$ for some point $p\notin H$ that contains at least one occurrence of $p$. The \textit{request sequence} of such a $(1,H)$-phase is the string $\rho$, its \textit{demand vector} is $u_p$, and its critical set is $\{p\}$.

For $\ell<k$, an \textit{$(\ell,H)$-multiphase} $Q$ is a sequence $P_1,\ldots,P_{w_{\ell+1}/w_{\ell}}$ of $(\ell,H)$ phases such that, if $\rho_i$ is the request sequence of each $P_i$, then for all $i$, $\rho_i$ is the longest prefix of $\rho_i\cdots \rho_{w_{\ell+1}/w_{\ell}}$ which is the request sequence of some $(\ell,H)$-phase. The \textit{request sequence} of the multiphase $Q$ is defined to be the string $\rho_1\cdots \rho_{w_{\ell+1}/w_{\ell}}$, and its \textit{demand vector} is defined to be $v=\sum_{i=1}^{w_{\ell+1}/w_{\ell}}v_i$, where $v_i$ is the demand vector of $P_i$. The \textit{critical set} of $Q$ is defined to be the set $\text{top}_{d_{\ell+1}-1}(v)$.

For $\ell>1$, an \textit{$(\ell,H)$-phase} $P$ is a pair $(Q^0,\mathcal{Q})$, where
\begin{itemize}
\item $Q^0$ is an $(\ell-1,H)$-multiphase, also called the \textit{explore part} of $P$. Let $\rho^0$ be the request sequence, $v^0$ be the demand vector, and $S$ be the critical set of $Q^0$.
\item $\mathcal{Q}$ is the set $\{Q^p\mid p\in S\}$, where each $Q^p$ is an $(\ell-1,H\cup\{p\})$-multiphase such that, if $\rho^p$ is the request sequence of $Q^p$, then the following conditions are satisfied.
\begin{enumerate}
\item $\{\rho^p\mid p\in S\}$ is a prefix chain. Let $\rho'$ be the longest string in this prefix chain.
\item $\rho^0$ is the longest prefix of $\rho^0\rho'$ which is the request sequence of some $(\ell-1,H)$-multiphase.
\item For each $p\in S$, $\rho^p$ is the longest prefix of $\rho'$ which is the request sequence of some $(\ell-1,H\cup\{p\})$-multiphase.
\end{enumerate}
$\mathcal{Q}$ is called the \textit{exploit part} of $P$.
\end{itemize}
The \textit{request sequence} of $P$ is defined to be the string $\rho^0\rho'$. The \textit{critical set} of $P$ is defined to be $S$, the critical set of the explore part of $P$. The \textit{demand vector} of $P$ is defined to be $v^0+\sum_{p\in S}v^p$, where $v^p$ is the demand vector of $Q^p$.
\end{definition}

The above definition has been tailor-made for the design and analysis of our online algorithm. Here, the demand vector of an $(\ell,H)$-(multi)phase can be thought of as a vector which, for each point in the metric space, indicates the utility of placing one of the $k-\ell$ heaviest servers at that point. The critical set is the set of an appropriate number of points that have the highest utility. Let us start by stating a few properties that follow from  Definition~\ref{def_phase} easily.

\begin{observation}\label{obs_non-empty}
For every $\ell\in\{1,\ldots,k\}$ (resp.\ $\ell\in\{1,\ldots,k-1\}$) and every set $H$ of at most $k-\ell$ points, the request sequence of every $(\ell,H)$-phase (resp.\ $(\ell,H)$-multiphase) is non-empty.
\end{observation}

The above observation is obvious from the fact that a $(1,H)$-phase is non-empty by definition, and from the recursive construction of (multi)phases given by Definition~\ref{def_phase}.

\begin{observation}\label{obs_demand_H}
For every $\ell\in\{1,\ldots,k\}$ (resp.\ $\ell\in\{1,\ldots,k-1\}$), every set $H$ of at most $k-\ell$ points, and every $p\in H$, the demand vector $v$ of every $(\ell,H)$-phase (resp.\ $(\ell,H)$-multiphase) satisfies $v[p]=0$.
\end{observation}

The above observation follows by a straightforward induction on the structure of (multi)phases. The next lemma gives an exact value of the $1$-norm of the demand vector of a (multi)phase.

\begin{restatable}{lemma}{lemeqsize}\label{lem_eq_size}
For every $\ell\in\{1,\ldots,k\}$ (resp.\ $\ell\in\{1,\ldots,k-1\}$) and every set $H$ of at most $k-\ell$ points, the demand vector $v$ of every $(\ell,H)$-phase (resp.\ $(\ell,H)$-multiphase) satisfies $|v|=(w_{\ell}/w_1)\cdot\prod_{i=1}^{\ell}d_i$ (resp.\ $|v|=(w_{\ell+1}/w_1)\cdot\prod_{i=1}^{\ell}d_i$).
\end{restatable}

The proof follows easily from Definition~\ref{def_phase}, and is deferred to Appendix~\ref{appendix_phase}. The next lemma states that if a string is the request sequence of an $(\ell,H)$-(multi)phase, then it has a unique hierarchical decomposition that serves as a witness of it being the request sequence of an $(\ell,H)$-(multi)phase. This is important to avoid any ambiguity in the presentation of our subsequent claims.

\begin{restatable}{lemma}{lemunamb}\label{lem_unambiguous}
For every $\ell\in\{1,\ldots,k\}$ (resp.\ $\ell\in\{1,\ldots,k-1\}$), every set $H$ of at most $k-\ell$ points and every string $\rho$ of requests, there exists at most one $(\ell,H)$-phase (resp.\ $(\ell,H)$-multiphase) whose request sequence is $\rho$.
\end{restatable}

The proof of the above lemma has been deferred to Appendix~\ref{appendix_phase}.

\section{Analysis of the Optimum Cost}\label{sec_opt}

In this section, we prove a lower bound on the cost of the optimum solution to the request sequence of any $(k,\emptyset)$-phase. We start by a few natural definitions.

\begin{definition}\label{def_conf}
A \textit{server configuration} is a $k$-tuple of points. The $\ell$'th point in a configuration $C$ is denoted by $C[\ell]$. Given a sequence of requests $\rho=r_1\cdots r_m$, a \textit{solution} to $\rho$ is a sequence $\mathcal{C}=[C_0,C_1,\ldots,C_m]$ of configurations such that for each $i\in\{1,\ldots,m\}$ there exists $\ell\in\{1,\ldots,k\}$ such that $r_i=C_i[\ell]$. The \textit{cost} of such a solution is $\sum_{i=1}^m\sum_{\ell=1}^kw_{\ell}\cdot\mathbb{I}[C_{i-1}[\ell]\neq C_i[\ell]]$. The cost of a minimum cost solution of a request sequence $\rho$ is denoted by $\text{OPT}(\rho)$.
\end{definition}

\begin{definition}\label{def_substring}
Given a sequence of requests $\rho=r_1\cdots r_m$ and $0\leq i\leq j\leq m$, $\rho[i,j]$ denotes the sequence of requests $r_{i+1}\cdots r_j$. If $\mathcal{C}=[C_0,C_1,\ldots,C_m]$ is a sequence of configurations, then $\mathcal{C}[i,j]$ denotes the sequence of configurations $[C_i,\ldots,C_j]$.
\end{definition}

\begin{observation}\label{obs_cost}
For every sequence of requests $\rho$ of length $m$, numbers $i,j$ such that $0\leq i\leq j\leq m$, and solution $\mathcal{C}$ to $\rho$, $\mathcal{C}[i,j]$ is a solution to $\rho[i,j]$. Hence $\text{OPT}(\rho)\geq\text{OPT}(\rho[i,j])$.
\end{observation}

\begin{definition} \label{def_restriction}
Let $P$ be an $(\ell,H)$-phase for an arbitrary $\ell\in\{2,\ldots,k\}$ and subset $H$ of at most $k-\ell$ points, having explore part $Q^0$, critical set $S$, exploit part $\mathcal{Q}=\{Q^p\mid p\in S\}$ (where $Q^p$ is an $(\ell-1,H\cup\{p\})$-multiphase) and request sequence $\rho=r_1\cdots r_m$. Recall that $Q^0$ and each $Q^p$ is a sequence of $(w_{\ell}/w_{\ell-1})$ $(\ell-1,H')$-phases (where $H'=H$ in case of $Q^0$ and $H'=H\cup\{p\}$ in case of $Q^p$). Let $P'$ be any of these phases. Observe that the request sequence of $P'$ is $\rho[i,j]$ for some $i,j$. Let $\mathcal{C}$ be a solution to $\rho$. Then we call the solution $\mathcal{C}[i,j]$ to $\rho[i,j]$ the \textit{restriction} of $\mathcal{C}$ to $P'$.
\end{definition}

\begin{definition}\label{def_l_active}
A sequence $\mathcal{C}=[C_0,C_1,\ldots,C_m]$ of configurations is said to be $\ell$-\textit{active} for $\ell\in\{1,\ldots,k\}$ if for all $\ell'\in\{\ell+1,\ldots,k\}$, we have $C_0[\ell']=C_1[\ell']=\cdots=C_m[\ell']$ (that is, no server heavier than the $\ell$'th lightest server moves).
\end{definition}

Our plan is to inductively prove a lower bound on the cost of a solution to the request sequence of any $(\ell,H)$-phase $P$ for every $\ell$, which finally gives us the desired bound for $(k,\emptyset)$-phases. If the optimum solution ever moves the $\ell$'th lightest server or any heavier server, then it already pays a lot. So assume that the optimum solution moves only the lightest $\ell-1$ servers, that is, it is $(\ell-1)$-active. Then we plan to prove that the solution must pay a sufficient cost on one of the $(\ell-1,H')$-multiphases within $P$ (which could be the explore part or in the exploit part of $P$). We call that multiphase the hard multiphase of $P$. The formal definition of a hard multiphase is as follows.

\begin{definition}\label{def_hard_multiphase}
Let $P$ be an $(\ell,H)$-phase for an arbitrary $\ell\in\{2,\ldots,k\}$ and subset $H$ of at most $k-\ell$ points, having explore part $Q^0$, critical set $S$, and exploit part $\mathcal{Q}=\{Q^p\mid p\in S\}$, where $Q^p$ is an $(\ell-1,H\cup\{p\})$-multiphase. Let $\mathcal{C}=[C_0,\ldots,C_m]$ be an $(\ell-1)$-active solution to the request sequence of $P$ (so that $C_0[\ell]=\cdots=C_m[\ell]$). Then the \textit{hard multiphase} of $P$ for $\mathcal{C}$ is defined to be $Q^0$ if $C_0[\ell]\notin S$, and $Q^{C_0[\ell]}$ if $C_0[\ell]\in S$.
\end{definition}

Our lower bound on the cost of a solution to the request sequence of any $(\ell,H)$-phase is on the condition that the heaviest $k-\ell$ servers are not occupying some crucial points associated with the phase in the beginning itself. If this is not the case, then we say that the solution is contaminated. The formal definition is as follows.

\begin{definition}\label{def_contamination}
Let $P$ be an $(\ell,H)$-phase for an arbitrary $\ell\in\{1,\ldots,k\}$ and subset $H$ of at most $k-\ell$ points. Let $\mathcal{C}=[C_0,\ldots,C_m]$ be an arbitrary $(\ell-1)$-active solution to the request sequence of $P$, and let $\ell'\in\{\ell+1,\ldots,k\}$ (so that $C_0[\ell']=\cdots=C_m[\ell']$). We say that $\mathcal{C}$ is $\ell'$-\textit{contaminated} for $P$ if one of the following conditions is satisfied.
\begin{itemize}
\item $\ell=1$ and $v[C_0[\ell']]=1$, where $v$ is the demand vector of $P$. (This condition is equivalent to $\ell=1$ and the critical set of $P$ being $\{C_0[\ell']\}$. The condition has been stated in such a way that it remains the same in the analysis of the generalized $k$-server too.)
\item $\ell>1$ and, out of the $w_{\ell}/w_{\ell-1}$ phases that make up the hard multiphase of $P$ for $\mathcal{C}$, at least a $2^{-(\ell'-\ell+3)}\cdot(w_{\ell}/w_{\ell-1})$ phases $P'$ are such that the restriction of $\mathcal{C}$ to $P'$ is ($(\ell-2)$-active, and) $\ell'$-contaminated for $P'$.
\end{itemize}
We say that $\mathcal{C}$ is \textit{contaminated} for $P$ if $\mathcal{C}$ is $\ell'$-contaminated for $P$ for some $\ell'\in\{\ell+1,\ldots,k\}$.
\end{definition}

Consider an $(\ell,H)$-phase $P$ and a non-contaminated $(\ell-1)$-active solution to the request sequence of $P$. As stated earlier, we plan to prove a lower bound on the cost of the solution by proving that it already pays a lot while serving the hard multiphase of $P$. This hard multiphase is made up of several $(\ell-1,H')$-phases. To inductively get a lower bound on the cost of the solution on any such subphase, the solution is required to remain non-contaminated for the subphase too. The following lemma gives an upper bound on the number of subphases for which the solution becomes contaminating due to the location of the $\ell$'th lightest server.

\begin{restatable}{lemma}{lemlcont}\label{lem_l_cont}
Let $P$ be an $(\ell,H)$-phase for an arbitrary $\ell\in\{2,\ldots,k\}$ and subset $H$ of at most $k-\ell$ points. Let $\mathcal{C}=[C_0,\ldots,C_m]$ be an arbitrary $(\ell-1)$-active solution to the request sequence of $P$. Then, out of the $w_{\ell}/w_{\ell-1}$ phases that make up the \textit{hard multiphase} of $P$ for $\mathcal{C}$, the number of phases $P'$ such that the restriction of $\mathcal{C}$ to $P'$ is $\ell$-contaminated for $P'$, is at most $w_{\ell}/(8w_{\ell-1})$.
\end{restatable}

The proof of the above lemma is deferred to Appendix~\ref{appendix_opt}. Using this lemma, we prove the following lemma which gives a lower bound on the cost of a non-contaminating solution to the request sequence of an $(\ell,H)$-phase.

\begin{lemma}\label{lem_opt_rec}
Let $P$ be an $(\ell,H)$-phase for an arbitrary $\ell\in\{2,\ldots,k\}$ and subset $H$ of at most $k-\ell$ points. Let $\mathcal{C}=[C_0,\ldots,C_m]$ be an arbitrary solution to the request sequence of $P$ such that $\mathcal{C}$ is not contaminated for $P$. Then the cost of $\mathcal{C}$ is at least $w_{\ell}/2^{\ell}$.
\end{lemma}

\begin{proof}
If the solution $\mathcal{C}$ is not an $(\ell-1)$-active solution, then it must move some server other than the $\ell-1$ lightest servers, so its cost is at least $w_{\ell}$. Therefore, let us assume that $\mathcal{C}$ is $(\ell-1)$-active. We prove that the solution must already pay at least $w_{\ell}/2^{\ell}$ while serving requests in the hard multiphase of $P$ for $\mathcal{C}$, say $Q$. Let $P_1,\ldots,P_{w_{\ell}/w_{\ell-1}}$ be the sequence of $(\ell-1,H')$ phases that makes up $Q$. By Lemma~\ref{lem_l_cont}, among $P_1,\ldots,P_{w_{\ell}/w_{\ell-1}}$, at most $w_{\ell}/(8w_{\ell-1})$ phases are such that the restriction of $\mathcal{C}$ to them is $\ell$-contaminated.  Moreover, we are given that for every $\ell'\in\{\ell+1,\ldots,k\}$, $\mathcal{C}$ is not $\ell'$-contaminated for $P$. By Definition~\ref{def_contamination}, at most $w_{\ell}/(2^{\ell'-\ell+3}\cdot w_{\ell-1})$ of the phases $P_1,\ldots,P_{w_{\ell}/w_{\ell-1}}$ are such that the restriction of $\mathcal{C}$ to them is $\ell'$-contaminated. Summing over $\ell'\in\{\ell,\ldots,k\}$, we get that, among $P_1,\ldots,P_{w_{\ell}/w_{\ell-1}}$, at most $w_{\ell}/(4w_{\ell-1})$ phases are such that the restriction of $\mathcal{C}$ to them is $\ell'$-contaminated for some $\ell' \in \{\ell,\ldots,k\}$. So, among $P_1,\ldots,P_{w_{\ell}/w_{\ell-1}}$, at least $(3w_{\ell})/(4w_{\ell-1})$ phases are such that the restriction of $\mathcal{C}$ to them is not contaminated for them. Call such phases \textit{costly}.

We prove the claim by induction on $\ell$. Suppose $\ell = 2$. Recall that, by the definition of a $(1,H')$-multiphase, each $P_i$ is a $(1,H')$-phase, and therefore, its critical set is $\{p_i\}$ for some point $p_i\notin H'$. Let $\rho_i$ denote the request sequence of $P_i$. Then $\rho_i$ is a string over the set $H'\cup\{p_i\}$ with at least one occurrence of $p_i$. Moreover, for each $i$, $p_i\neq p_{i+1}$, otherwise $\rho_i$ will not be the longest prefix of $\rho_i\rho_{i+1}\cdots\rho_{w_2/w_1}$ which is the request sequence of a $(1,H')$-phase; $\rho_i\rho_{i+1}$ would qualify to be a longer such prefix. By definition of contamination, for each costly phase $P_i$, none of the servers except the lightest server can serve requests to the critical point $p_i$, which means all such requests must be served by the lightest server. By a simple counting argument, among the costly phases, there are at least $w_2/(4w_1)$ pairs of consecutive phases. For each such pair $(P_i,P_{i+1})$, the cheapest server must move from $p_i$ to $p_{i+1}$, incurring a cost of $w_1$. Thus the total cost incurred by the solution is at least $w_2/(4w_1) \times w_1 = w_2/4$.

Next, let us consider the inductive case, where $\ell>2$. Suppose the claim holds for $\ell-1$. By induction hypothesis, in each of the at least $(3w_{\ell})/(4w_{\ell-1})$ costly phases in $Q$, any solution must incur a cost of $w_{\ell-1}/2^{\ell-1}$. Thus, the cost incurred by any solution is at least $((3w_{\ell})/(4w_{\ell-1}))\times (w_{\ell-1}/2^{\ell-1}) = (3w_{\ell})/2^{\ell+1} \geq w_{\ell}/2^{\ell}$.
\end{proof}

\begin{corollary}\label{cor_opt}
Let $P$ be an arbitrary $(k,\emptyset)$-phase. Then the cost of every solution to the request sequence of $P$ is at least $w_k/2^k$.
\end{corollary}

\begin{proof}
Follows from Lemma~\ref{lem_opt_rec} because the non-contamination condition required for $\mathcal{C}$ is vacuously true for $\ell=k$.
\end{proof}

\section{Recursive Construction and Analysis of Strategies}\label{sec_strategy}

In this section, we show inductively that the request sequence of an $(\ell,H)$-phase (resp.\ $(\ell,H)$-multiphase) can be served by a randomized online algorithm with expected cost at most $c_{\ell}\cdot w_{\ell}$ (resp.\ $c_{\ell}\cdot w_{\ell+1}$) for an appropriately defined constant $c_{\ell}$, provided the algorithm's heaviest $k-\ell$ servers cover all points in $H$ initially. We call such an algorithm an $\ell$-phase strategy (resp.\ $\ell$-multiphase strategy). In particular, a $k$-phase strategy serves the request sequence of a $(k,\emptyset)$-phase, starting from an arbitrary server configuration.

We begin by defining the constant $c_{\ell}$ for every $\ell$, followed by the rigorous definition of $\ell$-(multi)phase strategy. We use  $h:\mathbb{N}\cup\{0\}\longrightarrow\mathbb{R}$ to denote the \textit{harmonic function}, defined as $h(n)=1+1/2+\cdots+1/n$.

\begin{definition}\label{def_c}
The sequence $c_1,c_2,\ldots$ of constants is defined by the recurrence $c_1=1$ and for all $\ell>1$, $c_{\ell}=(1+h(d_{\ell}-1))\cdot c_{\ell-1}+2h(d_{\ell}-1)$.
\end{definition}

Unrolling the above recurrence, it is easy to see that $c_{\ell}$ given by the following observation.

\begin{observation}\label{obs_c}
$c_{\ell}=3\cdot\prod_{\ell'=2}^{\ell}(1+h(d_{\ell'}-1))-2$.
\end{observation}

\begin{definition}\label{def_alg}
For $\ell\in\{1,\ldots,k\}$ (resp.\ for $\ell\in\{1,\ldots,k-1\}$), an \textit{$\ell$-phase strategy} (resp.\ \textit{$\ell$-multiphase strategy}) is a randomized online algorithm $\mathcal{A}_{\ell}$ (resp.\ $\mathcal{A}'_{\ell}$) that satisfies the following specifications.
\begin{enumerate}
\item $\mathcal{A}_{\ell}$ (resp.\ $\mathcal{A}'_{\ell}$) takes as initial input a set $H$ of at most $k-\ell$ points and an initial configuration $C_0$ such that $H\subseteq\{C_0[\ell+1],\ldots,C_0[k]\}$.% It accesses
\item $\mathcal{A}_{\ell}$ (resp.\ $\mathcal{A}'_{\ell}$) takes as online input an arbitrary sequence $\rho$ of requests.
\item $\mathcal{A}_{\ell}$ (resp.\ $\mathcal{A}'_{\ell}$) serves requests by moving only the lightest $\ell$ servers. %\textcolor{red}{Is this important?}
\item If $\rho$ is not the request sequence of any $(\ell,H)$-phase (resp.\ $(\ell,H)$-multiphase) but $\rho$ has the request sequence of some $(\ell,H)$-phase (resp.\ $(\ell,H)$-multiphase) as a prefix, then
\begin{itemize}
\item $\mathcal{A}_{\ell}$ (resp.\ $\mathcal{A}'_{\ell}$) serves the longest such prefix, say $\rho^*$, and terminates.
\item Let $P$ denote the unique $(\ell,H)$-phase (resp.\ $(\ell,H)$-multiphase) whose request sequence is $\rho^*$ (uniqueness is guaranteed by Lemma~\ref{lem_unambiguous}). Upon termination, $\mathcal{A}_{\ell}$ (resp.\ $\mathcal{A}'_{\ell}$) returns the demand vector of $P$.
\end{itemize}
\item Else, $\mathcal{A}_{\ell}$ (resp.\ $\mathcal{A}'_{\ell}$) serves the whole of $\rho$ and waits for more requests.
\item The expected cost of $\mathcal{A}_{\ell}$ (resp.\ $\mathcal{A}'_{\ell}$) is at most $c_{\ell}\cdot w_{\ell}$ (resp.\ $c_{\ell}\cdot w_{\ell+1}$).
\end{enumerate}
\end{definition}

Note that the above specification states that, even though an $\ell$-(multi)phase strategy is a randomized algorithm, the prefix of a given request sequence served by the $\ell$-(multi)phase strategy is deterministic.

Our goal is to show the existence of a $k$-phase strategy. We construct such a strategy inductively using the next three lemmas.

\begin{restatable}{lemma}{lemalgbase}\label{lem_alg_base}
A $1$-phase strategy exists.
\end{restatable}

The design of a $1$-phase strategy is obvious from the definition of $(1,H)$-phase. The construction and the proof of the above lemma is deferred to Appendix~\ref{appendix_strategy}.

\begin{restatable}{lemma}{lemrecmultiphase}\label{lem_rec_multiphase}
For every $\ell\in\{1,\ldots,k-1\}$, if an $\ell$-phase strategy exists, then an $\ell$-multiphase strategy exists.
\end{restatable}

This design is straightforward too: given an $\ell$-phase strategy $\mathcal{A}_{\ell}$, the required $\ell$-multiphase strategy simply calls $\mathcal{A}_{\ell}$ $w_{\ell+1}/w_{\ell}$ times, adds up the vectors returned by these calls, and returns the result. The rigorous construction and proof of correctness are deferred to Appendix~\ref{appendix_strategy}.

\begin{restatable}{lemma}{lemrecphase}\label{lem_rec_phase}
For every $\ell\in\{2,\ldots,k\}$, if an $(\ell-1)$-multiphase strategy exists, then an $\ell$-phase strategy exists.
\end{restatable}

Let $\mathcal{A}'_{\ell-1}$ be an $(\ell-1)$-multiphase strategy. Consider the following algorithm, which we call $\mathcal{A}_{\ell}$, that takes as initial input a set of points $H$ of size at most $k-\ell$ and an initial configuration $C_0$ such that $H\subseteq\{C_0[\ell+1],\ldots,C_0[k]\}$.
\begin{enumerate}
\item Initialize $\mathcal{A}'_{\ell-1}$ with $H,C_0$.
\item Pass the requests from $\rho$ to $\mathcal{A}'_{\ell-1}$, and let it serve requests until it terminates. (We refer to this as the ``explore step'' in our analysis.)
\item Let $v^0$ be the demand vector returned by $\mathcal{A}'_{\ell-1}$. Set $S$ to be $\text{top}_{d_{\ell}-1}(v^0)$.
\item Sample $\pi:\{1,\ldots,d_{\ell}-1\}\longrightarrow S$, a uniformly random permutation of $S$. Set $f$ to $1$.
\item For each $p\in S$, initialize an instance of $\mathcal{A}'_{\ell-1}$ with $H\cup\{p\},C_0^p$, where $C_0^p$ is the configuration $C_0$ except that $C_0^p[\ell]=p$. Call this instance $\mathcal{A}'_{\ell-1}[p]$.
\item For each next request $r$ in the yet unserved suffix of $\rho$: (We refer to this as the ``exploit step'' in our analysis.)
\begin{enumerate}
\item Pass $r$ to each $\mathcal{A}'_{\ell-1}[p]$ that hasn't yet terminated.
\item For each $\mathcal{A}'_{\ell-1}[p]$ that terminates as a result, save the vector returned by it as $v^p$.
\item If $\mathcal{A}'_{\ell-1}[\pi(f)]$ just terminated, set $f$ to be the smallest number $f'$ in $\{1,\ldots,d_{\ell}-1\}$ such that $\mathcal{A}'_{\ell-1}[\pi(f')]$ hasn't yet terminated. If no such $f'$ exists, terminate and return $v^0+\sum_{p\in S}v^p$.
\item Move servers so that the configuration is the same as the current configuration of $\mathcal{A}'_{\ell-1}[\pi(f)]$. (This ensures that the request $r$ is served.)
\end{enumerate}
\end{enumerate}

Informally, $\mathcal{A}_{\ell}$ works as follows. Using $\mathcal{A}'_{\ell-1}$, it serves the longest prefix of the $\rho$ which is the request sequence of some $(\ell-1,H)$-multiphase $Q^0$. Once this call terminates returning the demand vector $v^0$ of $Q^0$, the critical set $S$ of $Q^0$ is computed. Thereafter, $\mathcal{A}_{\ell}$ exploits parallelism and randomness as follows. It creates $|S|$ copies of $\mathcal{A}'_{\ell-1}$, one for each $p\in S$, initialized with the set $H\cup\{p\}$. $\mathcal{A}_{\ell}$ forwards requests to each of these copies, one request at a time, and observes how each of them serves the request. $\mathcal{A}_{\ell}$, however, follows the decisions of exactly one copy at any time. Specifically, starting from a uniformly random copy of $\mathcal{A}'_{\ell-1}$, $\mathcal{A}_{\ell}$ follows a copy until it terminates, after which, $\mathcal{A}_{\ell}$ again chooses a uniformly random copy of $\mathcal{A}'_{\ell-1}$, among those who haven't yet terminated. This process continues as long as at least one copy of $\mathcal{A}'_{\ell-1}$ is running. Once the last copy terminates, $\mathcal{A}_{\ell}$ also terminates. Note that, except for the copy of $\mathcal{A}'_{\ell-1}$ that $\mathcal{A}_{\ell}$ is following, all copies of $\mathcal{A}'_{\ell-1}$ are merely ``running in imagination'' -- none of the server movements they attempt to perform become a reality. Our goal is to prove that $\mathcal{A}_{\ell}$ is an $\ell$-phase strategy.

Observe that all calls to $\mathcal{A}'_{\ell-1}$ made by $\mathcal{A}_{\ell}$ are initialized with a configuration that has the same positions of the heaviest $k-\ell$ servers as the initial configuration $C_0$. None of these calls move those servers, and $\mathcal{A}_{\ell}$ simply follows the movements of one of them at any point of time. Thus, $\mathcal{A}_{\ell}$ never moves the heaviest $k-\ell$ servers either. The next three claims state that $\mathcal{A}_{\ell}$ satisfies the last three properties in the specification of $(\ell,H)$-phase.

\begin{restatable}{claim}{claimserveone}
Suppose $\rho$ is not the request sequence of any $(\ell,H)$-phase but $\rho$ has the request sequence of some $(\ell,H)$-phase as a prefix. Then
\begin{itemize}
\item $\mathcal{A}_{\ell}$ serves the longest such prefix, say $\rho^*$, and terminates.
\item Let $P$ denote the unique $(\ell,H)$-phase whose request sequence is $\rho^*$. Upon termination, $\mathcal{A}_{\ell}$ returns the demand vector of $P$.
\end{itemize}
\end{restatable}

\begin{restatable}{claim}{claimservetwo}
Suppose $\mathcal{A}_{\ell}$ on input $\rho$ terminates. Then $\rho$ is not the request sequence of any $(\ell,H)$-phase but $\rho$ has the request sequence of some $(\ell,H)$-phase as a prefix.
\end{restatable}

The proofs of both the above claims are deferred to Appendix~\ref{appendix_strategy}.

\begin{claim}\label{claim_cost}
The expected cost the solution output by a run of $\mathcal{A}_{\ell}$ on any input is at most $c_{\ell}\cdot w_{\ell}$.
\end{claim}

\begin{proof}
Let the random variable $X^0$ denote the cost paid by the call to $\mathcal{A}'_{\ell-1}$ in the explore step of $\mathcal{A}_{\ell}$, and let $y^0$ denote the prefix of $\rho$ served by this call. For each $p\in S$ (computed in the third step of $\mathcal{A}_{\ell}$), let $y^p$ denote the prefix of $y^0\backslash\rho$ served by $\mathcal{A}'_{\ell-1}[p]$ (in imagination). Sort the set $S$ in non-increasing order of the length of $y^p$, and let $p_1,\ldots,p_{d_{\ell}-1}$ denote this sorted order. Let the random variable $X^j$ denote the (imaginary) cost paid by the call $\mathcal{A}'_{\ell-1}[p_j]$, and let $Z^j$ be the indicator random variable of the event that  $\mathcal{A}_{\ell}$ ever follows the decisions of $\mathcal{A}'_{\ell-1}[p_j]$, or equivalently, $\pi(f)=p_j$ for some time during the run of $\mathcal{A}_{\ell}$.

For every $j$ such that $Z^j=1$, the (real) cost paid by $\mathcal{A}_{\ell}$ while following the decisions of $\mathcal{A}'_{\ell-1}[p_j]$ is bounded from above by the (imaginary) cost paid by $\mathcal{A}'_{\ell-1}[p_j]$, that is, $X^j$. In addition, just before $\mathcal{A}_{\ell}$ begins to follow $\mathcal{A}'_{\ell-1}[p_j]$, $\mathcal{A}_{\ell}$ must pay a cost of at most $w_1+\cdots+w_{\ell}$ to match the positions of its cheapest $\ell$ servers with the positions of the corresponding servers of $\mathcal{A}'_{\ell-1}[p_j]$. Thus, the total cost paid by $\mathcal{A}_{\ell}$ is bounded from above by
\[X^0+\sum_{j=1}^{d_{\ell}-1}Z^j\cdot(w_1+\cdots+w_{\ell}+X^j)\leq X^0+\sum_{j=1}^{d_{\ell}-1}Z^j\cdot(2w_{\ell}+X^j)=X^0+2w_{\ell}\cdot\sum_{j=1}^{d_{\ell}-1}Z^j+\sum_{j=1}^{d_{\ell}-1}Z^j\cdot X^j\text{,}\]
where the inequality holds because we assumed that the weights are separated by a factor at least $2$.

Observe that $Z^j=1$ only if $p_j$ is placed before all of $p_1,\ldots,p_{j-1}$ in the random order $\pi$. Thus, $\mathbb{E}[Z^j]=\Pr[Z^j=1]\leq1/j$. By the definition of $(\ell-1)$-multiphase strategy, we have $\mathbb{E}[X^0]\leq c_{\ell-1}\cdot w_{\ell}$, and for each $j\in\{1,\ldots,d_{\ell}-1\}$, $\mathbb{E}[X^j]\leq c_{\ell-1}\cdot w_{\ell}$. Finally, note that the value of each $Z^j$ is determined by the random permutation $\pi$ while the value of $X^j$ is determined by the randomness internal to the run of $\mathcal{A}'_{\ell-1}[p^j]$. Thus, $Z^j$ and $X^j$ are independent random variables. Putting everything together, we get that the expected cost paid by $\mathcal{A}_{\ell}$ is bounded from above by
\begin{align*}
\mathbb{E}[X^0]+2w_{\ell}\cdot\sum_{j=1}^{d_{\ell}-1}\mathbb{E}[Z^j]+\sum_{j=1}^{d_{\ell}-1}\mathbb{E}[Z^j]\cdot\mathbb{E}[X^j] &\leq c_{\ell-1}\cdot w_{\ell}+2w_{\ell}\cdot\sum_{j=1}^{d_{\ell}-1}\frac{1}{j}+\sum_{j=1}^{d_{\ell}-1}\frac{1}{j}\cdot c_{\ell-1}\cdot w_{\ell}\\
&= ((1+h(d_{\ell}-1))\cdot c_{\ell-1}+2h(d_{\ell}-1))\cdot w_{\ell}\\
&= c_{\ell}\cdot w_{\ell}\text{,}
\end{align*}
where the last equality follows from Definition~\ref{def_c}.
\end{proof}

\begin{proof}[Proof of Lemma~\ref{lem_rec_phase}]
Follows from the definition of $\mathcal{A}_{\ell}$ and the last three claims.
\end{proof}

\begin{lemma}\label{lem_k_phase}
A $k$-phase strategy exists.
\end{lemma}

\begin{proof}
Follows from Lemma~\ref{lem_alg_base} and applications of Lemma~\ref{lem_rec_multiphase} and Lemma~\ref{lem_rec_phase}, alternated and repeated $k-1$ times, using appropriate values of $\ell$ each time.
\end{proof}

\section{Online Algorithm and its Analysis}\label{sec_alg}

Having demonstrated the existence of a $k$-phase strategy, the design of a competitive randomized algorithm for weighted $k$-server on uniform metrics becomes fairly straightforward.

\main*

\begin{proof}
The required algorithm is as follows: repeatedly call the $k$-phase strategy, whose existence is assured by Lemma~\ref{lem_k_phase}. Each call serves a non-empty sequence of requests, so all the requests are served. Suppose that on an arbitrary instance $\rho$, the $k$-phase strategy is called $m$ times. Then, except possibly the last call, all previous calls serve the request sequence of a $(k,\emptyset)$-phase. By Corollary~\ref{cor_opt}, irrespective of the initial configuration, the cost of an optimum solution for the request sequence of any  $(k,\emptyset)$-phase is at least $w_k/2^k$. Thus, the cost of an optimum solution to the instance is bounded as $\text{OPT}\geq(m-1)\cdot w_k/2^k$. On the other hand, the algorithm's cost is bounded as $\text{ALG}\leq m\cdot c_k\cdot w_k$. Thus, $\text{ALG}\leq2^k\cdot c_k\cdot\text{OPT}+c_kw_k$, which means that the algorithm is $(2^k\cdot c_k)$-competitive. From Observation~\ref{obs_c} and Definition~\ref{def_dl}, it is easy to see that $c_k$ is $\exp(\mathcal{O}(k^2))$, and hence, the competitive ratio is $\exp(\mathcal{O}(k^2))$. Note that all this holds under the assumption (which was made without loss of generality and stated in Section~\ref{sec_prelims}) that each $w_i$ is a multiple of $w_{i-1}$ greater than $w_{i-1}$. As remarked (and guaranteed by Lemma~\ref{lem_distort} in Appendix~\ref{appendix_prelims}), given arbitrary weights, they can be rounded up to satisfy the constraint with a loss of at most $2^{k-1}$ in the competitive ratio. This still keeps the competitive ratio bounded by $\exp(\mathcal{O}(k^2))$.
\end{proof}

\section{Concluding Remarks and Open Problems}

As a consequence of Theorem~\ref{thm_gks}, we have almost solved an open problem by Bansal et al.~\cite{BansalEKN_TALG23}, asking whether there exists an $\exp(k)$-competitive algorithm for the generalized $k$-server problem on weighted uniform metrics. We also have a reasonably comprehensive understanding of the competitive ratio of weighted $k$-server and generalized $k$-server: it is doubly exponential in $k$ for deterministic algorithms, while it is singly exponential in a polynomial of $k$ for randomized algorithms. The only gap left to close is between the $\Omega(2^k)$ and $\exp(\mathcal{O}(k^2))$ bounds on the randomized competitive ratio.

While the focus of competitive analysis is on proving information-theoretic results rather than computational efficiency, the latter cannot be ignored when algorithms are to be put into practice. Except for the randomized memoryless algorithms considered by Chiplunkar and Vishwanathan~\cite{ChiplunkarV_TAlg20}, all the algorithms for weighted $k$-server known so far need memory to be at least doubly exponential in $k$. It would therefore be interesting to understand the trade-off between memory usage and the competitive ratio of randomized algorithms for weighted $k$-server.

Since weighted $k$-server is reasonably understood on uniform metrics now, it makes sense to initiate a study of weighted $k$-server on classes of metrics beyond uniform metrics. For results to be non-trivial, the class of metrics under consideration must have metrics with arbitrarily large \textit{aspect ratio}, i.e., the ratio of the maximum distance to the minimum distance between points. The line metric and tree metrics might seem to be the simplest such classes, however, we believe designing algorithms for these metrics will already involve fundamentally new ideas. In particular, it is already known that randomized memoryless algorithms are not competitive on the line metric, even for the weighted $2$-server problem~\cite{ChrobakS_TCS04}.

\bibliographystyle{alpha}
\bibliography{references}

\appendix

\section{Illustrative Example}\label{appendix_example}

In this section, we present a concrete example to provide an intuitive visualization of the terms -- phase, multiphase, demand vector, critical set, etc.\ -- from Definition~\ref{def_phase}. Assume $k$ is at least $4$, and $w_2/w_1=5$. Although $d_2=16$ by Definition~\ref{def_dl}, we take $d_2=4$ in this example to reduce the clutter in the figure. The points in the metric space are $\texttt{a}, \texttt{b}, \text{etc.}$. Recall that demand vectors have a coordinate for each point in the metric space. A demand vector is shown as a set of key value pairs: for example, $\{\texttt{c}:1,\text{ }\texttt{d}:2,\text{ }\texttt{e}:1,\text{ }\texttt{f}:1\}$ denotes the vector that has value $2$ in the \texttt{d} coordinate, $1$ in each of $\texttt{c}$, $\texttt{e}$, $\texttt{f}$, and $0$ in all other coordinate directions.

\begin{figure}
\centering
\input{example.tex}
\caption{Illustrative example}
\label{fig_example}
\end{figure}

Refer to Figure~\ref{fig_example}, which shows a $(2,\{\texttt{a},\texttt{b}\})$-phase. The phase consists of an explore part and an exploit part. The explore part is a $(1,\{\texttt{a},\texttt{b}\})$-multiphase, and it consists of $w_2/w_1=5$ $(1,\{\texttt{a},\texttt{b}\})$-phases, each of which is enclosed in dashed boxes. The request sequence of each such phase is a string over $\texttt{a}$, $\texttt{b}$, and one other point, and the demand vector of the phase is a unit vector in the direction of the point other than  $\texttt{a}$, $\texttt{b}$. The request sequence is maximal in the sense that appending to it the first request from the next phase results in a string that doesn't qualify to be the request sequence of a $(1,\{\texttt{a},\texttt{b}\})$-phase. The demand vector of the explore part is the sum of the demand vectors of the $5$ $(1,\{\texttt{a},\texttt{b}\})$-phases in the explore part. The critical set $\{\texttt{c},\texttt{d},\texttt{e}\}$ is the set of $d_2-1=3$ points having the largest coordinate values in the demand vector, where ties are resolved alphabetically, due to which $\texttt{c}$ and $\texttt{e}$ are in the critical set, but $\texttt{f}$ isn't. Finally, observe that the request sequence of the explore part, which is simply the concatenation of the request sequences of the $(1,\{\texttt{a},\texttt{b}\})$-phases, is maximal in the sense that appending the first request from the exploit part results in a string that doesn't qualify to be the request sequence of a $(1,\{\texttt{a},\texttt{b}\})$-multiphase.

The exploit part of the $(2,\{\texttt{a},\texttt{b}\})$-phase consists of a $(1,\{\texttt{a},\texttt{b},\texttt{c}\})$-multiphase, a $(1,\{\texttt{a},\texttt{b},\texttt{d}\})$-multiphase, and a $(1,\{\texttt{a},\texttt{b},\texttt{e}\})$-multiphase, i.e., one multiphase for each point in the critical set. All these multiphases and the phases within them satisfy properties analogous to the explore part. Additionally, their request sequences form a prefix chain, and the longest of these request sequences is the request sequence of the exploit part. 

Finally, the demand vector of the $(2,\{\texttt{a},\texttt{b}\})$-phase shown in the figure is the sum of the demand vector of its explore part and the demand vectors of all multiphases in its exploit part, and the request sequence of the $(2,\{\texttt{a},\texttt{b}\})$-phase is the concatenation of the request sequences of its explore and exploit parts. Observe that, again, the request sequence of $(2,\{\texttt{a},\texttt{b}\})$-phase is maximal in the sense that appending the next request $\texttt{i}$ result in a string that doesn't qualify to be the request sequence of a $(2,\{\texttt{a},\texttt{b}\})$-phase.

\section{Missing Proofs}\label{app_missingproofs}

\subsection{Missing proof from Section~\ref{sec_prelims}}\label{appendix_prelims}

In this paper, we give an algorithm for the weighted $k$-server problem under the assumption that the weights satisfy the condition that each $w_i$ is an integer multiple of $w_{i-1}$ greater than $w_{i-1}$, and claim that there is no loss of generality in making this assumption. We formalize and prove this claim below. Let us call this version 
the \textit{weight-constrained} $k$-server problem.

\begin{lemma}\label{lem_distort}
If there exists a $c_k$-competitive randomized algorithm for  \textit{weight-constrained} $k$-server problem then there is a $(c_k\cdot 2^{k-1})$-competitive algorithm for the weighted $k$-server problem.
\end{lemma}

\begin{proof}
Suppose there is a $c_k$-competitive randomized algorithm $\mathcal{B}$ for  \textit{weight-constrained} $k$-server problem. We will now construct an algorithm for the weighted $k$-server problem with competitive ratio $c_k\cdot 2^{k-1}$.

Suppose the weights of the servers in the weighted $k$-server problem are $w_1,\ldots,w_k$, in non decreasing order. Define the weights $w'_1,\ldots,w'_k$ as $w'_1=w_1$, and for each $i>1$, $w'_i$ is the smallest multiple of $w'_{i-1}$ which is at least $\max(2w'_{i-1},w_i)$. Then we can inductively show that for each $i$, $w'_i\leq2^{i-1}w_i$. Indeed, the claim holds for $i=1$, and assuming that the claim holds for $i-1$, we have, $w'_i\leq\max(2w'_{i-1},w_i+w'_{i-1})\leq\max(2\cdot2^{i-2}w_{i-1},w_i+2^{i-2}w_{i-1})\leq2^{i-1}\cdot w_i$. Thus, $w_i\leq w'_i\leq2^{i-1}\cdot w_i$ for each $i$. Moreover, the weights $w'_1,\ldots,w'_k$ make the instance weight-constrained. The required weighted $k$-server algorithm is just algorithm $\mathcal{B}$ pretending as if the weights of the servers are $w'_1,\ldots w'_k$ instead of $w_1,\ldots, w_k$.

Given an arbitrary sequence of requests $\rho$, let $O$ and $O'$ be optimum solutions to $\rho$ with respect to the weights $w_1,\ldots, w_k$ and $w'_1,\ldots w'_k$ respectively. Let $w(O)$ denote the cost of $O$ with respect to weights $w_1,\ldots, w_k$, $w'(O)$ denote the cost of $O$ with respect to weights $w'_1,\ldots,w'_k$, and $w'(O')$ denote the cost of $O'$ with respect to weights $w'_1,\ldots, w'_k$. Then we have,
\[w'(O')\leq w'(O)\leq2^{k-1}\cdot w(O)\text{,}\]
where the first inequality holds because $O'$ is an optimum solution with respect to the weights $w'_1,\ldots, w'_k$, and $O$ is another solution, while the second inequality holds because $w'_i\leq2^{i-1}\cdot w_i\leq2^{k-1}\cdot w_i$ for each $i$.

Next, let $A$ denote the (random) solution output by the algorithm $\mathcal{B}$, and let $w(A)$ and $w'(A)$ respectively denote the expected cost of $A$ with respect to weights $w_1,\ldots, w_k$ and $w'_1,\ldots,w'_k$ respectively, so that $w(A)\leq w'(A)$. Since $\mathcal{B}$ is $c_k$ competitive, there exists a $c_0$ independent of $\rho$ such that $w'(A)\leq c_k\cdot w'(O')+c_0$. This implies,
\[w(A)\leq w'(A)\leq c_k\cdot w'(O')+c_0\leq c_k\cdot2^{k-1}\cdot w(O)+c_0\text{.}\]
Thus, we have a $(c_k\cdot 2^{k-1})$-competitive algorithm, as required.
\end{proof}

\subsection{Missing proofs from Section~\ref{sec_phase}}\label{appendix_phase}

\lemeqsize*

\begin{proof}
We will prove the claim by induction. Consider a $(1,H)$-phase. Recall that its demand vector is $u_p$ for some point $p$, and by definition $|u_p| = 1 = d_1$, as required.

Let us assume inductively that the claim holds for $(\ell,H)$-phases for an arbitrary $\ell\in\{1,\ldots,k-1\}$ and $H$, and prove it for an $(\ell,H)$-multiphase $Q$. Suppose $Q$ is the sequence $P_1,\ldots, P_{w_{\ell+1}/w_{\ell}}$ is of $(\ell,H)$-phases. Let $v_i$ denote the demand vector of $P_i$ and $v$ denote the demand vector of $Q$. Then
\[|v|=\sum_{i=1}^{w_{\ell+1}/w_{\ell}}|v_i|=\frac{w_{\ell+1}}{w_{\ell}}\times\frac{w_{\ell}}{w_1}\cdot\prod_{i=1}^{\ell}d_i=\frac{w_{\ell+1}}{w_1}\cdot\prod_{i=1}^{\ell}d_i\text{,}\]
as required, where the second equality follows from the inductive assumption.

Let us assume inductively that the claim holds for $(\ell-1,H')$-multiphases for an arbitrary $\ell\in\{2,\ldots,k\}$ and $H'$, and prove it for an $(\ell,H)$-phase $Q$. Let $Q^0$ and $\mathcal{Q}$ denote the explore and exploit parts of $P$ respectively, and let $S$ denote the critical set of $Q^0$. Recall that $Q^0$ is an $(\ell-1,H)$-multiphase, $|S|=d_{\ell}-1$, and $\mathcal{Q}$ contains an $(\ell-1,H\cup\{p\})$-multiphase $Q^p$ for each point $p\in S$. Let $v^0$, $v^p$, and $v$ denote the demand vectors of $Q^0$, each $Q^p$, and $P$ respectively. Then
\[|v|=|v^0|+\sum_{p\in S}|v^p|=d_{\ell}\times\frac{w_{\ell}}{w_1}\cdot\prod_{i=1}^{\ell-1}d_i=\frac{w_{\ell}}{w_1}\cdot\prod_{i=1}^{\ell}d_i\text{,}\]
as required, where the second equality follows from the inductive assumption.
\end{proof}

\lemunamb*

\begin{proof}
We will prove the claim by induction. First we will show that the lemma holds for $(1,H)$-phases. Since a $(1,H)$-phase is identified with its request sequence (provided the request sequence satisfies  the criteria in Definition~\ref{def_phase}), the claim holds.
    
Assume the lemma is true for $(\ell,H)$-phases. We will show it holds for $(\ell,H)$-multiphases as well. Let $Q^1 = (P^1_1,P^1_2,\ldots,P^1_{w_{\ell+1}/w_{\ell}})$ and $Q^2 = (P^2_1,P^2_2,\ldots,P^2_{w_{\ell+1}/w_{\ell}})$ be two $(\ell,H)$-multiphases with the same request sequence $\rho$. We denote the request sequence of $P^i_j$ as $\rho^i_j$. By Definition~\ref{def_phase}, both $\rho^1_1$ and $\rho^2_1$ are the longest prefixes of $\rho$ that are the request sequence of an $(\ell,H)$-phase, so, $\rho^1_1 = \rho^2_1$. Since the lemma is true for $(\ell,H)$-phases, this implies $P^1_1 = P^2_1$. Similarly, if $\rho^1_j = \rho^2_j$ for all $j\in \{1,\ldots,i-1\}$, then $\rho^1_i\cdots\rho^1_{w_{\ell+1}/w_{\ell}}=\rho^2_i\cdots\rho^2_{w_{\ell+1}/w_{\ell}}$. The  longest prefix of this string which is the request sequence of an $(\ell,H)$-phase is both $\rho^1_i$ as well as $\rho^2_i$, so $\rho^1_i = \rho^2_i$. Again, since the lemma is true for $(\ell,H)$-phases, $P^1_i = P^2_i$. Thus, by induction $P^1_i = P^2_i$ for all $i\in \{1,\ldots,w_{\ell+1}/w_{\ell}\}$. Thus $Q^1$ = $Q^2$, which implies the lemma is true for $(\ell,H)$-multiphases.

Assume the lemma is true for $(\ell-1,H)$-multiphases. We will show it is true for $(\ell,H)$-phases. Let $P_1 = (Q^0_1,\mathcal{Q}_1)$ and $P_2 = (Q^0_2,\mathcal{Q}_2)$ be two $(\ell,H)$-phases with the same request sequence $\rho$. We denote the request sequence of $Q^0_i$ as $\rho^0_i$ and that of $Q^p_i\in \mathcal{Q}_i$ as $\rho^p_i$. By Definition~\ref{def_phase}, $\rho^0_1,\rho^0_2$ are both the longest prefix of $\rho$ which is the request sequence of an $(\ell-1,H)$-multiphase. Thus $\rho^0_1=\rho^0_2 = \rho^0\text{ (say)}$.  Since the lemma is true for $(\ell-1,H)$-multiphases, $Q^0_1 = Q^0_2$. This implies that the critical sets of both $Q^0_1,Q^0_2$ are the same, say $S$. By Definition~\ref{def_phase}, for all $p \in S$, $\rho^p_1$ and $\rho^p_2$ are both the longest prefix of $\rho_0\backslash \rho$ which is the request sequence of an $(\ell-1,H\cup\{p\})$-multiphase. Thus, $\rho^p_1=\rho^p_2$, which implies $Q^p_1 = Q^p_2$ for all $p \in S$, by our inductive assumption. Thus $P_1 = P_2$, which implies the lemma is true for $(\ell,H)$-phases.
\end{proof}

\subsection{Missing proof from Section~\ref{sec_opt}}\label{appendix_opt}

We prove Lemma~\ref{lem_l_cont} here. A pre-requisite to proving it is Lemma~\ref{lem_cont_lbd}, which states that if a solution $\mathcal{C}=[C_0,\ldots,C_m]$ is $\ell'$-contaminated for a phase $P$, then the location $C_0[\ell']$ of the $\ell'$'th lightest server contributes significantly to $|v|$, where $v$ is the demand vector of $P$. To quantify this contribution, we need constants $F_{\ell,\ell'}$ defined as follows.

\begin{definition}\label{def_df}
The constants $F_{\ell,\ell'}$ for integers $\ell,\ell'$ such that $1\leq\ell\leq\ell'$ are defined as $F_{1,\ell'}=1$, and $F_{\ell,\ell'}=2^{\ell'-\ell+3}\cdot d_{\ell-1}\cdot F_{\ell-1,\ell'}$, if $\ell>1$.
\end{definition}

\begin{claim}\label{claim_df}
For every positive integer $\ell>1$, we have $d_{\ell}\geq2F_{\ell,\ell}$.
\end{claim}

\begin{proof}
Unrolling the recursion for $F_{\ell,\ell}$, we get,
\[F_{\ell,\ell} = 2^{\frac{5^{\ell-1}-1}{4}+(\ell-1)(\frac{\ell}{2}+1)}\text{.}\]
Thus $2F_{\ell,\ell}=2^{\frac{5^{\ell-1}-1}{4}+\frac{\ell^2+\ell}{2}}$. We want to show $d_{\ell}\geq 2F_{\ell,\ell}$, that is, $2^{5^{\ell-1}-1}\geq2^{\frac{5^{\ell-1}-1}{4}+\frac{\ell^2+\ell}{2}} $ for all $\ell\geq 2$. This is true if and only if, $5^{\ell-1}-1 \geq \frac{5^{\ell-1}-1}{4}+\frac{\ell^2+\ell}{2} $ for all $\ell \geq 2 $, that is, $5^{\ell-1}-1 \geq \frac{2\cdot(\ell^2+\ell)}{3}\text{ for all }\ell\geq 2$. Let $f(\ell) = 5^{\ell-1}-1$ and $g(\ell) = \frac{2\cdot(\ell^2+\ell)}{3} $. Observe that
\begin{align*}
f(2) &= g(2)\text{,}\\
f'(2) = 5\cdot\ln{5} &\geq \frac{10}{3} = g'(2)\text{, and}\\
f''(\ell) = 5^{\ell-1}\cdot(\ln{5})^2 &\geq \frac{4}{3} = g''(\ell)
\end{align*}
for all $\ell\in[2,\infty)$. Thus $f'(\ell)\geq g'(\ell)$ for all $\ell\in[2,\infty)$, and so,
$f(\ell)\geq g(\ell)$ for all $\ell\in[2,\infty)$, as required.
\end{proof}

\begin{lemma}\label{lem_cont_lbd}
Let $P$ be an $(\ell,H)$-phase for an arbitrary $\ell\in\{1,\ldots,k\}$ and subset $H$ of at most $k-\ell$ points, and let $v$ be the demand vector of $P$. Let $\mathcal{C}=[C_0,\ldots,C_m]$ be an arbitrary $(\ell-1)$-active solution to the request sequence of $P$, and let $\ell'\in\{\ell+1,\ldots,k\}$ (so that $C_0[\ell']=\cdots=C_m[\ell']$). If $\mathcal{C}$ is $\ell'$-contaminated for $P$, then
\[\frac{v[C_0[\ell']]}{|v|}\geq\frac{1}{d_{\ell}\cdot F_{\ell,\ell'}}\text{.}\]
\end{lemma}

\begin{proof}
We prove by induction on $\ell$. Suppose $\ell=1$. By Definition~\ref{def_contamination}, $v[C_0[\ell']]=1$, and by Definition~\ref{def_phase}, $|v|=1$. Since $d_1=F_{1,\ell'}=1$, the claim holds.

Suppose $\ell>1$. Consider the hard multiphase of $P$ for $\mathcal{C}$ and let $v'$ denote its demand vector. Recall that the hard multiphase is made up of a sequence $P_1,\ldots,P_{w_{\ell}/w_{\ell-1}}$ of $(\ell-1,H')$-phases for some fixed set $H'$. Let $v_i$ denote the demand vector of $P_i$, so that $v'=\sum_{i=1}^{w_{\ell}/w_{\ell-1}}v_i$. For each $i$ such that  the restriction of $\mathcal{C}$ to $P_i$ is ($(\ell-2)$-active, and) $\ell'$-contaminated for $P_i$, by induction hypothesis, we have,
\[v_i[C_0[\ell']]\geq\frac{|v_i|}{d_{\ell-1}\cdot F_{\ell-1,\ell'}}=\frac{(w_{\ell-1}/w_1)\times\prod_{j=1}^{\ell-1}d_j}{d_{\ell-1}\cdot F_{\ell-1,\ell'}}\text{,}\]
where the equality follows from Lemma~\ref{lem_eq_size}. By Definition~\ref{def_contamination}, at least $2^{-(\ell'-\ell+3)}\cdot(w_{\ell}/w_{\ell-1})$ indices $i$ are such that the restriction of $\mathcal{C}$ to $P_i$ is ($(\ell-2)$-active, and) $\ell'$-contaminated for $P_i$. Thus, we have,
\[v'[C_0[\ell']]\geq\frac{w_{\ell}/w_{\ell-1}}{2^{\ell'-\ell+3}}\cdot\frac{(w_{\ell-1}/w_1)\times\prod_{j=1}^{\ell-1}d_j}{d_{\ell-1}\cdot F_{\ell-1,\ell'}}=\frac{(w_{\ell}/w_1)\times\prod_{j=1}^{\ell-1}d_j}{2^{\ell'-\ell+3}\cdot d_{\ell-1}\cdot F_{\ell-1,\ell'}}=\frac{|v'|}{F_{\ell,\ell'}}\text{,}\]
where the last equality follows by Lemma~\ref{lem_eq_size} and Definition~\ref{def_df}. By Definition~\ref{def_phase}, $v$ is the sum of $d_{\ell}$ non-negative vectors, one of which is $v'$. Thus, $v[C_0[\ell']]\geq v'[C_0[\ell']]$. From Lemma~\ref{lem_eq_size}, we have $|v|=d_{\ell}\cdot|v'|$. Therefore, we get,
\[\frac{v[C_0[\ell']]}{|v|}\geq\frac{v'[C_0[\ell']]}{d_{\ell}\cdot|v'|}\geq\frac{1}{d_{\ell}\cdot F_{\ell,\ell'}}\text{,}\]
as required.
\end{proof}

\lemlcont*

\begin{proof}
Let $P = (Q^0, \mathcal{Q})$, where $Q^0$ and $\mathcal{Q}$ are respectively the explore and exploit parts of $P$. Let $S$ be the critical set of $P$, or equivalently, the critical set of $Q^0$.

First consider the case where $C_0[\ell]\notin S$, and therefore, $Q^0$ is the hard multiphase of $P$ for $\mathcal{C}$. Let $P_1,\ldots,P_{w_{\ell}/w_{\ell-1}}$ be the $(\ell-1,H)$-phases that make up the $(\ell-1,H)$-multiphase $Q^0$, and let $v_1,\ldots,v_{w_{\ell}/w_{\ell-1}}$ be their demand vectors. Then the demand vector $v'$ of $Q^0$ is given by $\sum_{i=1}^{w_{\ell}/w_{\ell-1}}v_i$. Let $B$ be the set of indices $i$ such that the restriction of $\mathcal{C}$ to $P_i$ is $\ell$-contaminated. For the sake of contradiction, assume that $|B|>w_{\ell}/(8w_{\ell-1})$. From Lemma~\ref{lem_cont_lbd}, for each $i\in B$, we have $v_i[C_0[\ell]]\geq|v_i|/(d_{\ell-1}\cdot F_{\ell-1,\ell})$. Therefore,
\[v'[C_0[\ell]] = \sum_{i=1}^{w_{\ell}/w_{\ell-1}}v_i[C_0[\ell]] \geq \sum_{i\in B}v_i[C_0[\ell]]\geq \sum_{i\in B}\frac{|v_i|}{d_{\ell-1}\cdot F_{\ell-1,\ell}}=\frac{|v'|}{w_{\ell}/w_{\ell-1}}\cdot\frac{|B|}{d_{\ell-1}\cdot F_{\ell-1,\ell}}\text{,}\]
where the last equality holds because by Lemma~\ref{lem_eq_size}, $|v'| = (w_{\ell}/w_{\ell-1})\cdot|v_i|$ for each $i$. Using the assumed lower bound of $|B|$, we have,
\[\frac{v'[C_0[\ell]]}{|v'|}\geq\frac{1}{w_{\ell}/w_{\ell-1}}\cdot\frac{w_{\ell}/8w_{\ell-1}}{d_{\ell-1}\cdot F_{\ell-1,\ell}} = \frac{1}{8d_{\ell-1}\cdot F_{\ell-1,\ell}}=\frac{1}{F_{\ell,\ell}}\geq\frac{2}{d_{\ell}}  \geq \frac{1}{d_{\ell}-1}\text{,}\]
where the second equality follows from Definition~\ref{def_df}, the second inequality follows from Claim~\ref{claim_df}, and the last inequality follows from the fact that $d_{\ell}\geq2$ for all $\ell>1$. But then the definition of critical set of a multiphase implies that $C_0[\ell]$ must be included in the critical set $S$ of $Q^0$, a contradiction. 

Next, consider the case where $C_0[\ell]\in S$, and therefore, the $(\ell-1, H\cup \{C_0[\ell]\})$-multiphase $Q^{C_0[\ell]}$ is the hard multiphase of $P$ for $\mathcal{C}$. $Q^{C_0[\ell]}$ consists of $w_{\ell}/w_{\ell-1}$ $(\ell-1, H\cup \{C_0[\ell]\})$-phases, and for each of them, by Observation~\ref{obs_demand_H}, its critical vector $v'$ satisfies $v'[C_0[\ell]]=0$. By (the contrapositive of) Lemma~\ref{lem_cont_lbd}, the restriction of $\mathcal{C}$ to any of these phases $P'$ is (either not $(\ell-2)$-active, or) not $\ell$-contaminated for $P'$. Thus, the claim holds.
\end{proof}

\subsection{Missing proofs from Section~\ref{sec_strategy}}\label{appendix_strategy}

Recall Definition~\ref{def_concat}. If a string $\rho_1$ is a prefix of a string $\rho_2$, then $\rho_1\backslash\rho_2$ denotes the result of chopping off the prefix $\rho_1$ from $\rho_2$, or in other words, the unique string $\rho_3$ such that $\rho_2=\rho_1\rho_3$. We will use this notation frequently in this section.

We first state and prove a lemma which will be useful in presenting the missing proofs. 
Suppose we are given a request sequence $\rho$ in an online manner and the task is to identify its longest prefix which is the request sequence of an $(\ell,H)$-(multi)phase (if such a prefix exists). The lemma states that, as soon as we find the smallest $m$ such that the first $m$ requests form the request sequence of some $(\ell,H)$-(multi)phase but the first $m+1$ requests don't, we can immediately stop and conclude that the former sequence is the one we need. It is impossible to append more requests to the latter sequence and make it the request sequence of an $(\ell,H)$-(multi)phase again.

\begin{lemma}\label{lem_prefix}
For every $\ell\in\{1,\ldots,k\}$ (resp.\ $\ell\in\{1,\ldots,k-1\}$) and every set $H$ of at most $k-\ell$ points, the following holds. Let $x$ and $z$ be the request sequence of two arbitrary $(\ell,H)$-phases (resp.\ $(\ell,H)$-multiphases), such that $x$ is a prefix of $z$. Then every prefix of $z$ longer than $x$ is also the request sequence of an $(\ell,H)$-phase (resp.\ $(\ell,H)$-multiphase).
\end{lemma}

\begin{proof}
We will prove the above lemma by induction. Let us first develop some notation for the same. $\mathcal{P}(\ell)$ and $\mathcal{P'}(\ell)$ are propositions defined as follows.

\begin{definition}[$\mathcal{P}(\ell)$ (resp.\ $\mathcal{P'}(\ell)$)]
For every set $H$ of at most $k-\ell$ points, the following holds. Let $x$ and $z$ be the request sequence of two arbitrary $(\ell,H)$-phases (resp.\ $(\ell,H)$-multiphases), such that $x$ is a prefix of $z$. Every prefix of $z$ longer than $x$ is also the request sequence of an $(\ell,H)$-phase (resp.\ $(\ell,H)$-multiphase).
\end{definition}

\begin{claim} 
$\mathcal{P}(1)$ is true.
\end{claim}

\begin{proof}
By Definition~\ref{def_phase}, $x$ and $z$ must be strings over $H\cup \{p\}$ and $H\cup \{q\}$ that contain at least one occurrence of $p$ and $q$ respectively, for some $p,q\notin H$. Since $x$ is a prefix of $z$, $p=q$. Thus, any prefix $y$ of $z$ which is longer than $x$ must also be a string over $H\cup \{p\}$. Hence, $y$ is also the request sequence of a $(1,H)$-phase. 
\end{proof}

\begin{claim}
For every $\ell\in\{1,\ldots,k-1\}$, $\mathcal{P}(\ell)$ implies $\mathcal{P'}(\ell)$.
\end{claim}

\begin{proof}
Consider the $(\ell,H)$-multiphases $Q^x$ and $Q^z$ with request sequences $x$ and $z$ respectively, and let $y$ be a prefix of $z$ which is longer than $x$. We will construct an $(\ell,H)$-multiphase $Q^y$ whose request sequence is $y$. Let the sequence of phases in $Q^x$ and $Q^z$ be $P^x_1,\ldots,P^x_m$ and $P^z_1,\ldots,P^z_m$ and their request sequences be $\rho^x_1,\ldots,\rho^x_{w_{\ell+1}/w_\ell}$ and $\rho^z_1,\ldots,\rho^z_{w_{\ell+1}/w_\ell}$. Recall that $\rho^x_i$ (resp.\ $\rho^z_i$) is the longest prefix of $\rho^x_i\cdots\rho^x_{w_{\ell+1}/w_\ell}$ (resp.\ $\rho^z_i\cdots\rho^z_{w_{\ell+1}/w_\ell}$) which is the request sequence of an $(\ell,H)$-phase.  We first show that for $1\leq i< w_{\ell+1}/w_\ell$, $\rho^x_i=\rho^z_i$. To prove this, assume the contrary. Let $j<w_{\ell+1}/w_\ell$ be the first index such that $\rho^x_j \neq \rho^z_j$. Then $\rho^x_j$ will be a strict prefix of $\rho^z_j$ or vice versa.

Assume that $\rho^x_j$ is a strict prefix of $\rho^z_j$. Given that $P^x_j$ is not the last phase in the multiphase, $\rho^x_{j+1}$ is not empty. Let its first request be $r$. Given $\mathcal{P}(\ell)$ is true, $\rho^x_jr$ is also the request sequence of an $(\ell,H)$-phase, and therefore, contradicts the maximality of $\rho^x_j$. On the other hand, assume that $\rho^z_j$ is a strict prefix of $\rho^x_j$. Since $x$ is a prefix of $z$, $\rho^x_j$ is a prefix of $\rho^z_j\cdots\rho^z_{w_{\ell+1}/w_\ell}$ . This contradicts the maximality of $\rho^z_j$. This proves that for $1\leq i<w_{\ell+1}/w_\ell$, $\rho^x_i=\rho^z_i$ and hence, $\rho^x_{w_{\ell+1}/w_\ell}$ is a prefix of $\rho^z_{w_{\ell+1}/w_\ell}$.

Note that $(\rho^x_1\cdots\rho^x_{(w_{\ell+1}/w_\ell)-1})\backslash y = \rho'$ (say) is a prefix of $\rho^z_{w_{\ell+1}/w_\ell}$ which is longer than $\rho^x_{w_{\ell+1}/w_\ell}$. Now, from $\mathcal{P}(\ell)$, we know that $\rho'$ is the request sequence of an $(\ell,H)$-phase, say $P^y_{w_{\ell+1}/w_\ell}$. Therefore, $Q^y=[P^x_1,P^x_2,\ldots,P^x_{(w_{\ell+1}/w_\ell)-1},P^y_{w_{\ell+1}/w_\ell}]$ is an $(\ell,H)$-multiphase with request sequence $y$.
\end{proof}

\begin{claim}
For every $\ell\in\{2,\ldots,k\}$, $\mathcal{P'}(\ell-1)$ implies $\mathcal{P}(\ell)$.
\end{claim}

\begin{proof} 
Assume $\mathcal{P'}(\ell-1)$ is true. Let $x$ and $z$ be the request sequences of two $(\ell,H)$-phases, $P_x$ and $P_z$ respectively, such that $x$ is a prefix of $z$. Suppose $y$ is a prefix of $z$ longer than $x$. We will construct an $(\ell,H)$-phase $P_y$, such that its request sequence is $y$.

Let $P_x = (Q_x,\mathcal{Q}_x)$ and $P_z = (Q_z,\mathcal{Q}_z)$. Let $x_0$ and $x_1$ be the request sequences of the explore and exploit parts of $P_x$, and let $z_0$ and $z_1$ be the request sequences of the explore and exploit parts of $P_z$, so that $x = x_0x_1$ and $z = z_0z_1 $. Let us first prove that $z_0$ is a prefix of $x$. Suppose instead that $x$ is a prefix of $z_0$. Since $x_0$ $z_0$ are the request sequences of two $(\ell-1,H)$-multiphases, by $\mathcal{P'}(\ell-1)$, $x$ is also the request sequence of an $(\ell-1,H)$-multiphase. Note that $x_0$ is the longest prefix of $x$ which is the request sequence of an $(\ell-1,H)$-multiphase, so $x=x_0$. However, this contradicts the non-emptiness of $x_1$ (Observation ~\ref{obs_non-empty}). Thus, $z_0$ must be a prefix of $x$. Recall that $x_0$ (resp.\ $z_0$) is the longest prefix of $x$ (resp.\ $z$) which is the request sequence of an $(\ell-1,H)$-multiphase. If $z_0$ is longer than $x_0$, then it contradicts the maximality of $x_0$. Similarly, if $x_0$ is longer than $z_0$, then it contradicts the maximality of $z_0$. Thus, $x_0=z_0$. By Lemma~\ref{lem_unambiguous}, $Q_x = Q_z$, which implies that their critical sets are the same, say $S$.

Recall that $\mathcal{Q}_x = \{Q^p_x\mid p\in S\}$ (resp.\ $\mathcal{Q}_z = \{Q^p_z\mid p\in S\}$) is the exploit part of $P_x$ (resp.\ $P_z$), where $Q^p_x$  (resp.\ $Q^p_z$) is the $(\ell-1,H\cup \{p\})$-multiphase with request sequence $x^p_1$ (resp.\ $z^p_1$) such that $x^p_1$ (resp.\ $z^p_1$) is the longest prefix of $x_1$ (resp.\ $z_1$) which is the request sequence of an $(\ell-1,H\cup \{p\})$-multiphase. The longest string among $\{x^p_1:p\in S\}$ (resp.\ $\{z^p_1:p\in S\}$) is $x_1$ (resp.\ $z_1$). Let $p^*\in S$ be such that $z_1$ is the request sequence of $Q^{p^*}_z\in \mathcal{Q}_z$. Note that $x_0\backslash y=z_0\backslash y$ is a prefix of $z_1$ which is longer than $x^{p^*}_1$. By $\mathcal{P'}(\ell-1)$, there exists a $(\ell-1,H\cup\{p^*\})$-multiphase $Q^{p^*}_y$, whose request sequence is $x_0\backslash y$. For every $p\in S$, let $Q^p_y$ be the $(\ell-1,H\cup \{p\})$-multiphase whose request sequence is the longest prefix of $x_0\backslash y$ which is the request sequence of an $(\ell-1,H\cup \{p\})$-multiphase (such a multiphase exists as the request sequence of $Q^p_x$ is a prefix of $x_0\backslash y$ and $Q^p_x$ is an $(\ell-1,H\cup \{p\})$-multiphase). Consider the phase $P_y = (Q_x,\{Q^p_y:p\in S\})$. The longest request sequence among the request sequences of $\{Q^p_y:p\in S\}$ is $x_0\backslash y$. Thus, the request sequence of $P_y$ is $y$. Therefore, there is an $(\ell,H)$-phase $P_y$ whose request sequence is $y$ which implies $\mathcal{P}(\ell)$ is true.
\end{proof}

The above three claims imply that the proposition $\mathcal{P}(\ell)$ (resp.\ $\mathcal{P}'(\ell)$) holds for all $\ell\in\{1,\ldots,k\}$ (resp.\ $\ell\in\{1,\ldots,k-1\}$), as required.
\end{proof}

\lemalgbase*

\begin{proof}
Recall that the request sequence of a $(1,H)$-phase is a string over the set $H\cup\{p\}$ with at least one occurrence of $p$, for some $p\notin H$. Consider the algorithm $\mathcal{A}_1$ that does the following, given a set $H$ of at most $k-1$ points and an initial configuration $C_0$ such that $H\subseteq\{C_0[2],\ldots,C_0[k]\}$. $\mathcal{A}_1$ ignores all requests until it finds the first occurrence of a point, say $p$, not in $H$. By assumption, the ignored requests are already covered by the algorithm's servers other than the lightest one. On the first request to $p$, $\mathcal{A}_1$ moves the lightest server to $p$. Thereafter, it ignores all requests in the set $H\cup\{p\}$, because they are already covered. As soon as it encounters a request not in $H\cup\{p\}$, $\mathcal{A}_1$ terminates without serving the request and returns the vector $u_p$. It is easy to see that $\mathcal{A}_1$ satisfies all the specifications of a $1$-phase.
\end{proof}

\lemrecmultiphase*

Let $\mathcal{A}_{\ell}$ be an $\ell$-phase strategy. The required $\ell$-multiphase strategy, which we call $\mathcal{A}'_{\ell}$, is fairly natural. Suppose a set $H$ of at most $k-\ell$ points and an initial configuration $C_0$ such that $H\subseteq\{C_0[\ell+1],\ldots,C_0[k]\}$ is given. $\mathcal{A}'_{\ell}$ initializes $v$ to the zero vector, and repeats the following sequence of steps $w_{\ell+1}/w_{\ell}$ times:
\begin{enumerate}
\item Call $\mathcal{A}_{\ell}$ with initial input $H,C_0$.
\item Pass the yet unserved suffix of $\rho$ to $\mathcal{A}_{\ell}$, and wait for it to terminate.
\item Let $v'$ be the vector returned by $\mathcal{A}_{\ell}$. Set $v\gets v+v'$.
\item Set $C_0$ to be the final configuration of $\mathcal{A}_{\ell}$.
\end{enumerate}
Once these iterations are over, $\mathcal{A}'_{\ell}$ terminates and returns the vector $v$.

By definition, since $\mathcal{A}_{\ell}$ be an $\ell$-phase strategy, each call to it moves only the lightest $\ell$ servers, and each call incurs expected cost of at most $c_{\ell}\cdot w_{\ell}$. Thus, $\mathcal{A}'_{\ell}$ moves only the lightest $\ell$ servers and its expected cost is at most $(w_{\ell+1}/w_{\ell})\times c_{\ell}\cdot w_{\ell}=c_{\ell}\cdot w_{\ell+1}$, as required. In the next two claims, we prove that $\mathcal{A}'_{\ell}$ satisfies the remaining two properties in the definition of $\ell$-multiphase strategy.

\begin{claim}
Suppose $\rho$ is not the request sequence of any $(\ell,H)$-multiphase but $\rho$ has the request sequence of some $(\ell,H)$-multiphase as a prefix. Then
\begin{itemize}
\item $\mathcal{A}'_{\ell}$ serves the longest such prefix, say $\rho^*$, and terminates.
\item Let $P$ denote the unique $(\ell,H)$-multiphase whose request sequence is $\rho^*$. Upon termination, $\mathcal{A}'_{\ell}$ returns the demand vector of $P$.
\end{itemize}
\end{claim}

\begin{proof}
We have $\rho^*=y_1\cdots y_{w_{\ell+1}/w_{\ell}}$, where $y_i$, for each $i\in\{1,\ldots,w_{\ell+1}/w_{\ell}-1\}$, is the longest prefix of $(y_1\cdots y_{i-1})\backslash\rho^*$ which is the request sequence of some $(\ell,H)$-phase, and $y_{w_{\ell+1}/w_{\ell}}$ is also the request sequence of some $(\ell,H)$-phase. Note that for $i<w_{\ell+1}/w_{\ell}$, $y_i$ is not the entire $(y_1\cdots y_{i-1})\backslash\rho^*$, since the subsequent $y_j$'s cannot be empty. Thus, $y_i$ must be the longest prefix of $(y_1\cdots y_{i-1})\backslash\rho$ which is the request sequence of some $(\ell,H)$-phase -- any longer prefix would contradict the maximality of $y_i$, possibly using Lemma~\ref{lem_prefix}. Moreover, $y_{w_{\ell+1}/w_{\ell}}$ is not the entire $(y_1\cdots y_{i-1})\backslash\rho$ (else $\rho^*=\rho$, making $\rho$ the request sequence of an $(\ell,H)$-multiphase), and $y_{w_{\ell+1}/w_{\ell}}$ must be the longest prefix of $(y_1\cdots y_{w_{\ell+1}/w_{\ell}-1})\backslash\rho$ which is the request sequence of some $(\ell,H)$-phase -- any longer such prefix can replace $y_{w_{\ell+1}/w_{\ell}}$ in $\rho^*$ to give a longer prefix of $\rho$ which is the request sequence of some $(\ell,H)$-multiphase. Thus, each $y_i$ must be the longest prefix of $(y_1\cdots y_{i-1})\backslash\rho$ which is the request sequence of some $(\ell,H)$-phase, and $y_i\neq(y_1\cdots y_{i-1})\backslash\rho$.

Consider the run of $\mathcal{A}'_{\ell}$ on $\rho$. The $i$'th call to $\mathcal{A}_{\ell}$ must serve exactly $y_i$, and return the demand vector of the unique $(\ell,H)$-phase whose request sequence is $y_i$. $\mathcal{A}'_{\ell}$ computes the sum of these demand vectors, which is the demand vector of $P$, by definition. Thus, $\mathcal{A}'_{\ell}$ serves $y_1\cdots y_{w_{\ell+1}/w_{\ell}}=\rho^*$ and returns the demand vector of $P$.
\end{proof}

\begin{claim}
Suppose $\mathcal{A}'_{\ell}$ on input $\rho$ terminates. Then $\rho$ is not the request sequence of any $(\ell,H)$-multiphase but $\rho$ has the request sequence of some $(\ell,H)$-multiphase as a prefix.
\end{claim}

\begin{proof}
Let $x$ be the prefix of $\rho$ which is served by $\mathcal{A}'_{\ell}$. Then $x=x_1\cdots x_{w_{\ell+1}/w_{\ell}}$, where $x_i$ is the string of requests served by the $i$'th call to $\mathcal{A}_{\ell}$. At the time the $i$'th call is made, the string or requests yet to be served is $(x_1\cdots x_{i-1})\backslash\rho$. Since $\mathcal{A}_{\ell}$ is an $\ell$-phase strategy, $x_i$ is the longest prefix of $(x_1\cdots x_{i-1})\backslash\rho$ which is the request sequence of some $(\ell,H)$-phase. Since $x$ is a prefix of $\rho$, $x_i$ is the longest prefix of $(x_1\cdots x_{i-1})\backslash x=x_i\cdots x_{w_{\ell+1}/w_{\ell}}$ which is the request sequence of some $(\ell,H)$-phase. Thus, $x=x_1\cdots x_{w_{\ell+1}/w_{\ell}}$ is the request sequence of an $(\ell,H)$-multiphase, which means $\rho$ has the request sequence of some $(\ell,H)$-multiphase as prefix. We are left to prove that $\rho$ itself is not the request sequence of an $(\ell,H)$-multiphase.

Suppose $\rho$ itself is the request sequence of an $(\ell,H)$-multiphase. Then $\rho=y_1\cdots y_{w_{\ell+1}/w_{\ell}}$, where $y_i$, for each $i\in\{1,\ldots,w_{\ell+1}/w_{\ell}-1\}$, is the longest prefix of $(y_1\cdots y_{i-1})\backslash\rho$ which is the request sequence of some $(\ell,H)$-phase, and $y_{w_{\ell+1}/w_{\ell}}$ is also the request sequence of some $(\ell,H)$-phase. But then this implies $y_i=x_i$ for each $i$, by induction on $i$. Thus, the last call to $\mathcal{A}_{\ell}$ got $y_{w_{\ell+1}/w_{\ell}}=x_{w_{\ell+1}/w_{\ell}}$, which is the request sequence of an $(\ell,H)$-phase, as its online input. Therefore, upon serving $x_{w_{\ell+1}/w_{\ell}}$, the last call $\mathcal{A}_{\ell}$ cannot terminate. This is a contradiction to the assumption that $\mathcal{A}'_{\ell}$ on input $\rho$ terminates.
\end{proof}

\begin{proof}[Proof of Lemma~\ref{lem_rec_multiphase}]
Follows from the definition of $\mathcal{A}'_{\ell}$ and the last two claims.
\end{proof}

We are left to complete the proof of Lemma~\ref{lem_rec_phase} from Section~\ref{sec_strategy}, stated as follows. 

\lemrecphase*

We now present the deferred proofs of the claims that constitute the proof of the above lemma.

\claimserveone*

\begin{proof}
Let $P=(Q^0,\mathcal{Q})$ denote the unique $(\ell,H)$-phase whose request sequence is $\rho^*$. Let $S$ denote the critical set of $P$. The request sequence $y^0$ of $Q^0$ is the longest prefix of $\rho^*$ which is the request sequence of some $(\ell-1,H)$-multiphase. For each $p\in S$, let $y^p$ denote the longest prefix of $y^0\backslash\rho^*$ which is the request sequence of some $(\ell-1,H\cup\{p\})$-multiphase, say $Q^p$. Then $y^0\backslash\rho^*$ is the longest of these $y^p$'s, which means there exists a $p^*\in S$ such that $y^{p^*}=y^0\backslash\rho^*$. In particular, by Observation~\ref{obs_non-empty}, $y^{p^*}$ is non-empty. This implies that $y^0$ is also the longest prefix of $\rho$ which is the request sequence of some $(\ell-1,H)$-multiphase -- any longer prefix would contradict the maximality of $y^0$, possibly using Lemma~\ref{lem_prefix}. Moreover, $y^{p^*}=y^0\backslash\rho^*$ is not the entire $y^0\backslash\rho$ (else $\rho^*=\rho$, making $\rho$ the request sequence of an $(\ell,H)$-phase). For every $p\in S$ such that $y^p$ is shorter than $y^{p^*}$, $y^p$ must also be the longest prefix of $y^0\backslash\rho$ which is the request sequence of some $(\ell-1,H\cup\{p\})$-multiphase -- any longer prefix would contradict the maximality of $y^p$, possibly using Lemma~\ref{lem_prefix}. For every $p\in S$ such that $y^p=y^{p^*}$, $y^p$ must also be the longest prefix of $y^0\backslash\rho$ which is the request sequence of some $(\ell-1,H\cup\{p\})$-multiphase -- any longer prefix would contradict the maximality of $\rho^*$. Thus, for each $p\in S$, $y^p$ is the longest prefix of $y^0\backslash\rho$ which is the request sequence of some $(\ell-1,H\cup\{p\})$-multiphase.

Consider the run of $\mathcal{A}_{\ell}$ on $\rho$. The call to $\mathcal{A}'_{\ell-1}$ in the explore step must serve exactly $y^0$, and return the demand vector $v^0$ of $Q^0$. In the next step, the critical set $S$ of $Q^0$ is computed. Consider the run of $\mathcal{A}'_{\ell-1}[p]$ for an arbitrary $p\in S$ in the exploit step. It must, in imagination, serve exactly $y^p$ and return the demand vector $v^p$ of $Q^p$. Thus, $\mathcal{A}_{\ell}$ serves exactly the longest of all $y^p$'s, namely $y^{p^*}=y^0\backslash\rho^*$, in the exploit step. $\mathcal{A}_{\ell}$ computes and returns $v^0+\sum_{p\in S}v^p$, which is the demand vector of $P$, by definition. Thus, $\mathcal{A}'_{\ell}$ serves $\rho^*$ and returns the demand vector of $P$.
\end{proof}

\claimservetwo*

\begin{proof}
Let $x$ be the prefix of $\rho$ which is served by $\mathcal{A}_{\ell}$. Then $x=x^0x'$, where $x^0$ is the string of requests served by call to $\mathcal{A}'_{\ell-1}$ in the explore step, and $x'$ is the string of requests served in the exploit step. Since $\mathcal{A}'_{\ell-1}$ is an $(\ell-1)$-multiphase strategy, $x^0$ is the longest prefix of $\rho$ which is the request sequence of some $(\ell-1,H)$-multiphase, say $Q^0$. The call to $\mathcal{A}'_{\ell-1}$ in the explore step terminates and returns the demand vector of $Q^0$, from which the critical set $S$ of $Q$ is computed by $\mathcal{A}_{\ell}$. For each $p\in S$, the instance $\mathcal{A}'_{\ell-1}[p]$ created in the exploit step serves, in imagination, the longest prefix $x^p$ of $x^0\backslash\rho$ which is the request sequence of some $(\ell-1,H\cup\{p\})$-multiphase, say $Q^p$. Since we are given that $\mathcal{A}_{\ell}$ on input $\rho$ terminates, the instance $\mathcal{A}'_{\ell-1}[p]$ on input $x^0\backslash\rho$ must terminate, and therefore, $x^p$ is well defined. Since each $x^p$ is a prefix of $x^0\backslash\rho$, $\{x^p\mid p\in S\}$ is a prefix chain. $\mathcal{A}_{\ell}$, in reality, in its exploit step, serves the longest of the $x^p$'s, and therefore, that string equals $x'$. Note that $x^0$ is the longest prefix of $x$ which is the request sequence of some $(\ell-1,H)$-multiphase, because $x$ is a prefix of $\rho$. Similarly, each $x^p$ is the longest prefix of $x'$ which is the request sequence of some $(\ell-1,H\cup\{p\})$-multiphase, because $x'$ is a prefix of $x^0\backslash\rho$. Thus, $x$ is the request sequence of the $(\ell,H)$-phase $P=(Q^0,\mathcal{Q})$, where $\mathcal{Q}=\{Q^p\mid p\in S\}$. This means $\rho$ has the request sequence of some $(\ell,H)$-phase as prefix. We are left to prove that $\rho$ itself is not the request sequence of an $(\ell,H)$-phase.

Suppose $\rho$ itself is the request sequence of an $(\ell,H)$-phase $P'$. Then $\rho=y^0y'$, where $y^0$ and $y'$ are the request sequences of the explore and exploit parts of $P$ respectively. By definition, $y^0$ is the longest prefix of $\rho$ which is the request sequence of some $(\ell-1,H)$-multiphase. Thus, $y_0=x_0$, and by Lemma~\ref{lem_unambiguous}, the explore part of $P'$ is also $Q^0$. This means that there exists $p'\in S$, the critical set of $Q^0$, such that $y'=y_0\backslash\rho=x_0\backslash\rho$ is the request sequence of some $(\ell-1,H\cup\{p'\})$-multiphase. Consider the instance $\mathcal{A}'_{\ell-1}[p']$ created in the exploit step. It is the $(\ell-1)$-multiphase strategy $\mathcal{A}'_{\ell-1}$ made to run on $y'$, which is itself the request sequence of an $(\ell-1,H\cup\{p'\})$-multiphase. Therefore, $\mathcal{A}'_{\ell-1}[p']$ must serve the whole $y'$ and wait for more requests. Thus, $\mathcal{A}_{\ell}$ on input $\rho$ cannot terminate.
\end{proof}

\section{Generalized $k$-Server on Weighted Uniform Metrics}\label{app_gks}

In this section, we show how the techniques used in the design of the weighted $k$-server algorithm can be lifted to the generalized $k$-server on uniform metrics, thus proving Theorem~\ref{thm_gks}.

\thmgks*

The proof of this theorem is analogous to the proof of Theorem~\ref{thm_main}, once we have proven claims analogous to Lemma~\ref{lem_k_phase} and Corollary~\ref{cor_opt}. We sketch the proofs of these claims by enlisting the minor modifications that need to be done. As before, we assume $w_1\leq w_2 \leq \cdots \leq w_k$.

\subsection{Modifications to Section~\ref{sec_prelims}}

Let $M=\bigcup_{i=1}^kM_i$. Definitions \ref{def_concat}, \ref{def_one_hot}, and \ref{def_dl} remain the same as before. In addition, we need the following definitions.

\begin{definition}
Given a vector $v$ in the span of $\{u_p\mid p\in M\}$, we denote its projection on the span of $\{u_p\mid p\in M_{\ell}\}$ by $v\vert_{\ell}$.
\end{definition}

\begin{definition}
A subset $H$ of $M$ is defined to be an \textit{$\ell$-configuration} if $H$ contains at most one point from each of $M_{\ell+1},\ldots,M_k$, and no point from any of $M_1,\ldots,M_{\ell}$.
\end{definition}

\begin{definition}
We say a set $H\subseteq M$ \textit{satisfies} the request $r$, if there exists $\ell\in\{1,\ldots,k\}$, such that $r[\ell]\in H$.
\end{definition}

\subsection{Modifications to Section~\ref{sec_phase}}

\begin{definition}[Analogous to Definition~\ref{def_phase}]
Let $\ell\in\{1,\ldots,k\}$ and let $H$ be an $\ell$-configuration.

A \textit{$(1,H)$-phase} is a sequence $\rho$ of requests, all of which are satisfied by $H\cup\{p\}$ for some point $p\in M_1$, and at least one of which 
is not satisfied by $H$.
Let $r$ be the first such request (which implies $r[1]=p$).
The \textit{request sequence} of such a $(1,H)$-phase is $\rho$, its \textit{demand vector} is $\sum_{\ell=1}^ku_{r[\ell]}$, and its critical set is $\{p\}$.

For $\ell<k$, an \textit{$(\ell,H)$-multiphase} $Q$,
its \textit{request sequence}, and its \textit{demand vector} $v$ are defined exactly as in Definition~\ref{def_phase}.
The \textit{critical set} of $Q$ is defined to be the set $\text{top}_{d_{\ell+1}-1}(v\vert_{l+1})$.

For $\ell>1$, an \textit{$(\ell,H)$-phase} $P$,
its \textit{explore} and \textit{exploit parts}, its \textit{request sequence}, \textit{demand vector}, and \textit{critical set} are all defined exactly as in Definition~\ref{def_phase}.
\end{definition}

From now on, when we refer to an $(\ell,H)$-phase or $(\ell,H)$-multiphase, we assume that $H$ is an $\ell$-configuration. Observations \ref{obs_non-empty} and \ref{obs_demand_H} remain true for this definition of phase.

\begin{lemma}[Analogous to Lemma~\ref{lem_eq_size}]
For every $\ell\in\{1,\ldots,k\}$ (resp.\ $\ell\in\{1,\ldots,k-1\}$) and every set $H$, the demand vector $v$ of every $(\ell,H)$-phase (resp.\ $(\ell,H)$-multiphase) satisfies, for all $\ell'\in\{1,\ldots,k\}$, $|(v\vert_{\ell'})|=(w_{\ell}/w_1)\cdot\prod_{i=1}^{\ell}d_i$ (resp.\ $|(v\vert_{\ell'})|=(w_{\ell+1}/w_1)\cdot\prod_{i=1}^{\ell}d_i$).
\end{lemma}

\begin{proof}
The proof is analogous to the proof of Lemma~\ref{lem_eq_size}.
\end{proof}

Note that the above lemma implies that if $v$ is the demand vector of a phase or a multiphase, then for every $\ell,\ell'\in[k], |(v\vert_{\ell})|=|(v\vert_{\ell'})|$. Lemmas \ref{lem_unambiguous} and \ref{lem_prefix} continue to hold, and their proofs remain unchanged.

\subsection{Modifications to Section~\ref{sec_opt}}

\begin{definition}[Analogous to Definition~\ref{def_conf}]
A \textit{server configuration} is a $k$-tuple of points, one from each $M_{\ell}$. The $\ell$'th point in a configuration $C$ is denoted by $C[\ell]$ (thus, $C[\ell]\in M_{\ell}$). Given a sequence of requests $\rho=r_1\cdots r_m$, a \textit{solution} to $\rho$ is a sequence $\mathcal{C}=[C_0,C_1,\ldots,C_m]$ of configurations such that for each $i\in\{1,\ldots,m\}$ there exists $\ell\in\{1,\ldots,k\}$ such that $r_i[\ell]=C_i[\ell]$. The \textit{cost} of such a solution is $\sum_{i=1}^m\sum_{\ell=1}^kw_{\ell}\cdot\mathbb{I}[C_{i-1}[\ell]\neq C_i[\ell]]$. The cost of a minimum cost solution of a request sequence $\rho$ is denoted by $\text{OPT}(\rho)$.
\end{definition}

Definitions \ref{def_substring}, \ref{def_restriction}, \ref{def_l_active}, \ref{def_hard_multiphase}, \ref{def_contamination} and Observation~\ref{obs_cost} remain unchanged except that we use the new definitions of phase and configuration. 

Definition~\ref{def_df} remains as it is, and Claim~\ref {claim_df} remains true. 

\begin{lemma}[Analogous to Lemma~\ref{lem_cont_lbd}]
Let $P$ be an $(\ell,H)$-phase for an arbitrary $\ell\in\{1,\ldots,k\}$ and an $\ell$-configuration $H$, and let $v$ be the demand vector of $P$. Let $\mathcal{C}=[C_0,\ldots,C_m]$ be an arbitrary $(\ell-1)$-active solution to the request sequence of $P$, and let $\ell'\in\{\ell+1,\ldots,k\}$ (so that $C_0[\ell']=\cdots=C_m[\ell']$). If $\mathcal{C}$ is $\ell'$-contaminated for $P$, then
\[\frac{v[C_0[\ell']]}{|(v\vert_{\ell'})|}\geq\frac{1}{d_{\ell}\cdot F_{\ell,\ell'}}\text{.}\]
\end{lemma}

\begin{proof}
The proof is analogous to the proof of Lemma~\ref{lem_cont_lbd} with the modification that $v$ is replaced by $v\vert_{\ell'}$.
\end{proof}

Lemmas \ref{lem_l_cont}, \ref{lem_opt_rec} and Corollary~\ref{cor_opt} remain true with analogous proofs.

\subsection{Modifications to Section~\ref{sec_strategy}}

Definition~\ref{def_c}, and Observation~\ref{obs_c} remain unchanged.

\begin{definition}[Analogous to Definition~\ref{def_alg}]
For $\ell\in\{1,\ldots,k\}$ (resp.\ for $\ell\in\{1,\ldots,k-1\}$), an \textit{$\ell$-phase strategy} (resp.\ \textit{$\ell$-multiphase strategy}) is a randomized online algorithm $\mathcal{A}_{\ell}$ (resp.\ $\mathcal{A}'_{\ell}$) that satisfies the following specifications.
\begin{enumerate}
\item $\mathcal{A}_{\ell}$ (resp.\ $\mathcal{A}'_{\ell}$) takes as initial input an $\ell$-configuration $H$ and an initial configuration $C_0$ such that $H\subseteq\{C_0[\ell+1],\ldots,C_0[k]\}$.% It accesses
\item $\mathcal{A}_{\ell}$ (resp.\ $\mathcal{A}'_{\ell}$) takes as online input an arbitrary sequence $\rho$ of requests.
\item $\mathcal{A}_{\ell}$ (resp.\ $\mathcal{A}'_{\ell}$) serves requests by moving only the servers in $M_1,\ldots,M_{\ell}$. %\textcolor{red}{Is this important?}
\item If $\rho$ is not the request sequence of any $(\ell,H)$-phase (resp.\ $(\ell,H)$-multiphase) but $\rho$ has the request sequence of some $(\ell,H)$-phase (resp.\ $(\ell,H)$-multiphase) as a prefix, then
\begin{itemize}
\item $\mathcal{A}_{\ell}$ (resp.\ $\mathcal{A}'_{\ell}$) serves the longest such prefix, say $\rho^*$, and terminates.
\item Let $P$ denote the unique $(\ell,H)$-phase (resp.\ $(\ell,H)$-multiphase) whose request sequence is $\rho^*$ (uniqueness is guaranteed by Lemma~\ref{lem_unambiguous}). Upon termination, $\mathcal{A}_{\ell}$ (resp.\ $\mathcal{A}'_{\ell}$) returns the demand vector of $P$.
\end{itemize}
\item Else, $\mathcal{A}_{\ell}$ (resp.\ $\mathcal{A}'_{\ell}$) serves the whole of $\rho$ and waits for more requests.
\item The expected cost of $\mathcal{A}_{\ell}$ (resp.\ $\mathcal{A}'_{\ell}$) is at most $c_{\ell}\cdot w_{\ell}$ (resp.\ $c_{\ell}\cdot w_{\ell+1}$).
\end{enumerate}
\end{definition}

Lemmas \ref{lem_alg_base}, and \ref{lem_rec_multiphase} remain true with similar proofs. In line 3 of the algorithm in Lemma~\ref{lem_rec_phase}, we set $S$ to $\text{top}_{d_{\ell}-1}(v\vert_{\ell})$. Lemma~\ref{lem_rec_phase} is proven exactly as before, thus implying Lemma~\ref{lem_k_phase}.

\end{document}

%% file: example.tex
\begin{center}
\rotatebox{90}{
\begin{tikzpicture}[scale=0.5]

% request sequence

\draw[<->] (-1,21) -- (40,21) node[midway, above] {input};

\node at (-0.5,19.5) {\texttt{\ldots}};
\node at (0.5,19.5) {\texttt{a}};
\node at (1.5,19.5) {\texttt{b}};
\node at (2.5,19.5) {\texttt{c}};
\node at (3.5,19.5) {\texttt{a}};
\node at (4.5,19.5) {\texttt{d}};
\node at (5.5,19.5) {\texttt{a}};
\node at (6.5,19.5) {\texttt{d}};
\node at (7.5,19.5) {\texttt{b}};
\node at (8.5,19.5) {\texttt{e}};
\node at (9.5,19.5) {\texttt{a}};
\node at (10.5,19.5) {\texttt{d}};
\node at (11.5,19.5) {\texttt{b}};
\node at (12.5,19.5) {\texttt{a}};
\node at (13.5,19.5) {\texttt{f}};
\node at (14.5,19.5) {\texttt{b}};
\node at (15.5,19.5) {\texttt{c}};
\node at (16.5,19.5) {\texttt{a}};
\node at (17.5,19.5) {\texttt{b}};
\node at (18.5,19.5) {\texttt{c}};
\node at (19.5,19.5) {\texttt{d}};
\node at (20.5,19.5) {\texttt{a}};
\node at (21.5,19.5) {\texttt{b}};
\node at (22.5,19.5) {\texttt{c}};
\node at (23.5,19.5) {\texttt{b}};
\node at (24.5,19.5) {\texttt{e}};
\node at (25.5,19.5) {\texttt{a}};
\node at (26.5,19.5) {\texttt{c}};
\node at (27.5,19.5) {\texttt{b}};
\node at (28.5,19.5) {\texttt{a}};
\node at (29.5,19.5) {\texttt{d}};
\node at (30.5,19.5) {\texttt{b}};
\node at (31.5,19.5) {\texttt{g}};
\node at (32.5,19.5) {\texttt{a}};
\node at (33.5,19.5) {\texttt{b}};
\node at (34.5,19.5) {\texttt{d}};
\node at (35.5,19.5) {\texttt{a}};
\node at (36.5,19.5) {\texttt{h}};
\node at (37.5,19.5) {\texttt{d}};
\node at (38.5,19.5) {\texttt{i}};
\node at (39.5,19.5) {\texttt{\ldots}};

% phase content without boxes

\node at (0.5,17.5) {\texttt{a}};
\node at (1.5,17.5) {\texttt{b}};
\node at (2.5,17.5) {\texttt{c}};
\node at (3.5,17.5) {\texttt{a}};
\node at (4.5,17.5) {\texttt{d}};
\node at (5.5,17.5) {\texttt{a}};
\node at (6.5,17.5) {\texttt{d}};
\node at (7.5,17.5) {\texttt{b}};
\node at (8.5,17.5) {\texttt{e}};
\node at (9.5,17.5) {\texttt{a}};
\node at (10.5,17.5) {\texttt{d}};
\node at (11.5,17.5) {\texttt{b}};
\node at (12.5,17.5) {\texttt{a}};
\node at (13.5,17.5) {\texttt{f}};
\node at (14.5,17.5) {\texttt{b}};
\node at (15.5,13.5) {\texttt{c}};
\node at (16.5,13.5) {\texttt{a}};
\node at (17.5,13.5) {\texttt{b}};
\node at (18.5,13.5) {\texttt{c}};
\node at (19.5,13.5) {\texttt{d}};
\node at (20.5,13.5) {\texttt{a}};
\node at (21.5,13.5) {\texttt{b}};
\node at (22.5,13.5) {\texttt{c}};
\node at (23.5,13.5) {\texttt{b}};
\node at (24.5,13.5) {\texttt{e}};
\node at (25.5,13.5) {\texttt{a}};
\node at (26.5,13.5) {\texttt{c}};
\node at (27.5,13.5) {\texttt{b}};
\node at (28.5,13.5) {\texttt{a}};
\node at (29.5,13.5) {\texttt{d}};
\node at (30.5,13.5) {\texttt{b}};
\node at (31.5,13.5) {\texttt{g}};
\node at (32.5,13.5) {\texttt{a}};
\node at (33.5,13.5) {\texttt{b}};
\node at (34.5,13.5) {\texttt{d}};
\node at (35.5,13.5) {\texttt{a}};

\node at (2,16.5) {\texttt{\{c:1\}}};
\node at (6,16.5) {\texttt{\{d:1\}}};
\node at (9,16.5) {\texttt{\{e:1\}}};
\node at (11.5,16.5) {\texttt{\{d:1\}}};
\node at (14,16.5) {\texttt{\{f:1\}}};
\node at (19.5,12.5) {\texttt{\{d:1\}}};
\node at (26.5,12.5) {\texttt{\{e:1\}}};
\node at (30,12.5) {\texttt{\{d:1\}}};
\node at (32.5,12.5) {\texttt{\{g:1\}}};
\node at (35,12.5) {\texttt{\{d:1\}}};

\node at (7.5,15.5) {\texttt{\{c:1, d:2, e:1, f:1\}}};
\node at (7.5,14.5) {\texttt{(1,\{a,b\})}-multiphase};

\node at (7.5,9) {critical set};
\node at (7.5,8) {$\text{top}_3($\texttt{\{c:1, d:2, e:1, f:1\}}$)$};
\node at (7.5,7) {$=$\texttt{\{c,d,e\}}};

\node at (25.5,11.5) {\texttt{\{d:3, e:1, g:1\}}};
\node at (25.5,10.5) {\texttt{(1,\{a,b,c\})}-multiphase};

\node at (15.5,9.5) {\texttt{c}};
\node at (16.5,9.5) {\texttt{a}};
\node at (17.5,9.5) {\texttt{b}};
\node at (18.5,9.5) {\texttt{c}};
\node at (19.5,9.5) {\texttt{d}};
\node at (20.5,9.5) {\texttt{a}};
\node at (21.5,9.5) {\texttt{b}};
\node at (22.5,9.5) {\texttt{c}};
\node at (23.5,9.5) {\texttt{b}};
\node at (24.5,9.5) {\texttt{e}};
\node at (25.5,9.5) {\texttt{a}};
\node at (26.5,9.5) {\texttt{c}};
\node at (27.5,9.5) {\texttt{b}};
\node at (28.5,9.5) {\texttt{a}};
\node at (29.5,9.5) {\texttt{d}};
\node at (30.5,9.5) {\texttt{b}};
\node at (31.5,9.5) {\texttt{g}};
\node at (32.5,9.5) {\texttt{a}};
\node at (33.5,9.5) {\texttt{b}};
\node at (34.5,9.5) {\texttt{d}};
\node at (35.5,9.5) {\texttt{a}};
\node at (36.5,9.5) {\texttt{h}};
\node at (37.5,9.5) {\texttt{d}};

\node at (19.5,8.5) {\texttt{\{c:1\}}};
\node at (25,8.5) {\texttt{\{e:1\}}};
\node at (28.5,8.5) {\texttt{\{c:1\}}};
\node at (33.5,8.5) {\texttt{\{g:1\}}};
\node at (37,8.5) {\texttt{\{h:1\}}};

\node at (26.5,7.5) {\texttt{\{c:2, e:1, g:1, h:1\}}};
\node at (26.5,6.5) {\texttt{(1,\{a,b,d\})}-multiphase};

\node at (15.5,5.5) {\texttt{c}};
\node at (16.5,5.5) {\texttt{a}};
\node at (17.5,5.5) {\texttt{b}};
\node at (18.5,5.5) {\texttt{c}};
\node at (19.5,5.5) {\texttt{d}};
\node at (20.5,5.5) {\texttt{a}};
\node at (21.5,5.5) {\texttt{b}};
\node at (22.5,5.5) {\texttt{c}};
\node at (23.5,5.5) {\texttt{b}};
\node at (24.5,5.5) {\texttt{e}};
\node at (25.5,5.5) {\texttt{a}};
\node at (26.5,5.5) {\texttt{c}};
\node at (27.5,5.5) {\texttt{b}};
\node at (28.5,5.5) {\texttt{a}};
\node at (29.5,5.5) {\texttt{d}};
\node at (30.5,5.5) {\texttt{b}};
\node at (31.5,5.5) {\texttt{g}};
\node at (32.5,5.5) {\texttt{a}};
\node at (33.5,5.5) {\texttt{b}};

\node at (17,4.5) {\texttt{\{c:1\}}};
\node at (20.5,4.5) {\texttt{\{d:1\}}};
\node at (25.5,4.5) {\texttt{\{c:1\}}};
\node at (30,4.5) {\texttt{\{d:1\}}};
\node at (32.5,4.5) {\texttt{\{g:1\}}};

\node at (24.5,3.5) {\texttt{\{c:2, d:2, g:1\}}};
\node at (24.5,2.5) {\texttt{(1,\{a,b,e\})}-multiphase};

\draw[|-|] (0,1) -- (15,1) node[midway, below] {explore part};
\draw[|-|] (15,1) -- (38,1) node[midway, below] {exploit part};

\node at (19,-0.5) {\texttt{\{c:5, d:7, e:3, f:1, g:3, h:1\}}};
\node at (19,-1.5) {\texttt{(2,\{a,b\})}-phase};

% boxes

\draw[dashed,thick] (0,-2) rectangle (38,18);

\draw[thick] (0,14) rectangle (15,18);
\draw[dashed,thick] (0,16) -- (15,16);
\draw[dashed,thick] (4,16) -- (4,18);
\draw[dashed,thick] (8,16) -- (8,18);
\draw[dashed,thick] (10,16) -- (10,18);
\draw[dashed,thick] (13,16) -- (13,18);

\draw[thick] (15,2) rectangle (38,14);
%\draw[thick] (15,2) -- (15,14);
\draw[thick] (15,6) -- (38,6);
\draw[thick] (15,10) -- (38,10);
\draw[thick] (34,2) -- (34,6);
\draw[thick] (36,10) -- (36,14);
\draw[dashed,thick] (15,4) -- (34,4);
\draw[dashed,thick] (15,8) -- (38,8);
\draw[dashed,thick] (15,12) -- (36,12);
\draw[dashed,thick] (24,12) -- (24,14);
\draw[dashed,thick] (29,12) -- (29,14);
\draw[dashed,thick] (31,12) -- (31,14);
\draw[dashed,thick] (34,12) -- (34,14);
\draw[dashed,thick] (24,8) -- (24,10);
\draw[dashed,thick] (26,8) -- (26,10);
\draw[dashed,thick] (31,8) -- (31,10);
\draw[dashed,thick] (36,8) -- (36,10);
\draw[dashed,thick] (19,4) -- (19,6);
\draw[dashed,thick] (22,4) -- (22,6);
\draw[dashed,thick] (29,4) -- (29,6);
\draw[dashed,thick] (31,4) -- (31,6);

\end{tikzpicture}
}
\end{center}